\documentclass[12pt,oneside]{amsart}

\usepackage{amsmath, amssymb, amsthm, amscd, url}

\newcommand{\kne}{\mathrm{KN}}

\newcommand{\pd}{{\partial}}
\newcommand{\frl}{\mathfrak{F}}
\newcommand{\frid}{\mathfrak{I}}
\newcommand{\mcl}{\mathfrak{L}}

\newcommand{\lspan}{\mathrm{span\,}}
\newcommand{\wea}{\mathfrak{W}}

\newcommand{\fla}{\mathbf{A}}
\newcommand{\flb}{\mathbf{B}}
\newcommand{\flm}{\mathbf{M}}
\newcommand{\fln}{\mathbf{N}}
\newcommand{\flz}{\mathbf{Z}}

\newcommand{\fnt}{F}
\newcommand{\hlt}{H}

\newcommand{\idz}{\mathfrak{Z}}
\newcommand{\fanl}{\mathcal{A}}

\newcommand{\ufl}{\mathfrak{R}}
\newcommand{\fik}{\mathbb{K}}
\newcommand{\qalg}{\mathcal{Q}}
\newcommand{\ipol}{\mathcal{I}}

\newcommand{\wfla}{\hat{\fla}}
\newcommand{\wflb}{\hat{\flb}}

\newcommand{\ple}{M}
\newcommand{\qle}{N}
\newcommand{\cle}{A}
\newcommand{\dle}{B}

\newcommand{\nmf}{n_1}
\newcommand{\nmg}{n_2}
\newcommand{\mtf}{v}

\newcommand{\wflm}{\hat{\flm}}
\newcommand{\wfln}{\hat{\fln}}

\newcommand{\al}{{\alpha}}
\newcommand{\be}{{\beta}}
\newcommand{\la}{{\lambda}}

\newcommand{\mat}{\mathbf{S}}
\newcommand{\ub}{\underbrace}

\newcommand{\cl}{\colon}

\newcommand{\CE}{\mathcal{E}}
\newcommand{\ess}{\mathcal{E}'}

\newcommand{\zp}{\mathbb{Z}_{\ge 0}}
\newcommand{\zsp}{\mathbb{Z}_{>0}}

\newcommand{\ad}{{\rm ad\,}}

\newcommand{\mg}{\mathfrak{g}}
\newcommand{\vf}{\varphi}
\newcommand{\Com}{\mathbb{C}}

\newtheorem{theorem}{Theorem}
\newtheorem{proposition}{Proposition}
\newtheorem{statement}{Statement}
\newtheorem{lemma}{Lemma}

\theoremstyle{definition}
\newtheorem{example}{Example}
\newtheorem{remark}{Remark}

\begin{document}



\title[Prolongation Lie algebras and 
multicomponent Landau-Lifshitz systems]
{Infinite-dimensional prolongation Lie algebras \\ 
and multicomponent Landau-Lifshitz systems \\
associated with higher genus curves}
\date{}

\author{Sergey Igonin}
\address{S.~Igonin,
Department of Mathematics, Utrecht University, P.O. Box 80010, 3508 TA Utrecht,
the Netherlands}
\email{s-igonin@yandex.ru}

\author{Johan van de Leur}
\address{J.~van~de~Leur, 
Department of Mathematics, Utrecht University, P.O. Box 80010, 3508 TA Utrecht,
the Netherlands}
\email{J.W.vandeLeur@uu.nl}

\author{Gianni Manno}
\address{G.~Manno \\
Mathematisches Institut, Fakult\"at f\"ur Mathematik und Informatik, 
Friedrich-Schiller-Universit\"at Jena, D-07737 Jena, Germany}
\address{Universit\`a degli Studi di Milano-Bicocca, Dipartimento di Matematica e Applicazioni, Via Cozzi 53, 20125 Milano, Italy}
\email{gianni.manno@unimib.it}

\author{Vladimir Trushkov}
\address{V.~Trushkov \\
University of Pereslavl, Sovetskaya 2, 152020 Pereslavl-Zalessky, Yaroslavl region, Russia}
\email{vvtrushkov@yandex.ru}

\begin{abstract}
The Wahlquist-Estabrook prolongation method 
constructs for some PDEs a Lie algebra 
that is responsible for Lax pairs and B\"acklund transformations of certain type. 
We present some general properties of Wahlquist-Estabrook algebras  
for $(1+1)$-dimensional evolution PDEs and compute 
this algebra for the $n$-component Landau-Lifshitz system of Golubchik and Sokolov for any $n\ge 3$.  

We prove that 
the resulting algebra is isomorphic to the direct sum of a $2$-dimensional abelian Lie algebra 
and an infinite-dimensional Lie algebra $L(n)$ of certain matrix-valued functions 
on an algebraic curve of genus $1+(n-3)2^{n-2}$. 
This curve was used by Golubchik, Sokolov, Skrypnyk, Holod 
in constructions of Lax pairs. 
Also, we find a presentation for the algebra $L(n)$
in terms of a finite number of generators and relations. 
These results help to obtain a partial answer 
to the problem of classification of multicomponent Landau-Lifshitz systems 
with respect to B\"acklund transformations.

Furthermore, we construct a family of integrable evolution PDEs 
that are connected with the $n$-component Landau-Lifshitz system by Miura type transformations 
parametrized by the above-mentioned curve. 
Some solutions of these PDEs are described. 
\end{abstract}

\maketitle

\section{Introduction}

\subsection{Motivation for the studied problem and a summary of the results}
\label{motiv}

In the last 30 years, it has been relatively well understood how 
to obtain integrable PDEs from some infinite-dimensional Lie algebras
(see, e.g.,~\cite{adler,djkm,dickey,ds,feher,mll,gds,jimbo,reyman-semenov,
skr,skr-jmp} and references therein).
We study the inverse problem: given a 
PDE\footnote{A ``PDE'' means a ``system of partial differential equations''.}, 
how to determine whether this PDE is related to an infinite-dimensional
Lie algebra and how to construct the corresponding Lie algebra?

A partial answer to this question is provided by 
the so-called Wahlquist-Estabrook prolongation method~\cite{dodd,nonl,pirani79,Prol}.  
For a given $(1+1)$-dimensional evolution PDE, 
this method constructs a Lie algebra in terms of generators and relations. 
It is called the \emph{Wahlquist-Estabrook algebra} of the PDE 
(WE algebra for short). The method is applicable also 
to some non-evolution PDEs (see, e.g.,~\cite{finley-k2,pirani79}).
  
The construction of the WE algebra for a PDE uses only the PDE itself. 
Here the PDE does not have to be integrable. 
When the WE algebra turns out to be infinite-dimensional, 
this is usually a serious indication that 
the PDE possesses some integrability properties. 
  
Before describing the results of this paper, 
we would like to recall some known applications of WE algebras. 
Any matrix representation of the WE algebra of a PDE determines a 
zero-curvature representation (ZCR) for this PDE. 
(For $(1+1)$-dimensional PDEs, 
the notion of ZCR is essentially equivalent to that of Lax pair.)
Vector field representations of the WE algebra often lead to B\"acklund 
transformations. 
Computing the structure of WE algebras for PDEs, 
one can get many interesting infinite-dimensional Lie algebras 
(see, e.g.,~\cite{kdv1,finley-k2,schief,cfa,ll} and references therein).

Using some generalization of WE algebras, 
one obtains powerful necessary conditions for two given PDEs 
to be connected by a B\"acklund transformation (BT for short)~\cite{cfa,igon-mpi,cfg-2009}. 
For example, the following result has been proved recently in~\cite{igon-mpi} by means of this theory. 
For any $e_1,e_2,e_3\in\Com$, 
consider the Krichever-Novikov equation 
\begin{equation}
\label{kne}
  \kne(e_1,e_2,e_3)=\left\{
 u_t=u_{xxx}-\frac32\frac{u_{xx}^2}{u_x}+
 \frac{(u-e_1)(u-e_2)(u-e_3)}{u_x},\qquad u=u(x,t)\right\},
\end{equation}
and the algebraic curve 
$C(e_1,e_2,e_3)=\Big\{(z,y)\in\Com^2\ \Big|\ 
y^2=(z-e_1)(z-e_2)(z-e_3)\Big\}$.
\begin{proposition}[\cite{igon-mpi}]\label{knprop}
Let $e_1,e_2,e_3,e'_1,e'_2,e'_3\in\Com$ be such that 
$e_i\neq e_j$ and $e'_i\neq e'_j$ for all $i\neq j$. 

If the curve $C(e_1,e_2,e_3)$ is not birationally equivalent to 
the curve $C(e'_1,e'_2,e'_3)$, 
then the equation $\kne(e_1,e_2,e_3)$ is not connected with the equation $\kne(e'_1,e'_2,e'_3)$ by any B\"acklund transformation. 

Also, if $e_1\neq e_2\neq e_3\neq e_1$, then $\kne(e_1,e_2,e_3)$ is not connected with the KdV equation by any BT. 
\end{proposition}
Similar results are proved in~\cite{igon-mpi} for the Landau-Lifshitz and 
nonlinear Schr\"odinger equations as well.

BTs of Miura type (differential substitutions) for~\eqref{kne} were studied 
in~\cite{svin-sok83}. 
According to~\cite{svin-sok83}, the equation $\kne(e_1,e_2,e_3)$
is connected with the KdV equation by a BT of Miura type 
iff $e_i=e_j$ for some $i\neq j$.

The papers~\cite{igon-mpi,cfg-2009} and 
Proposition~\ref{knprop} consider the most general class of BTs, 
which is much larger than the class of 
BTs of Miura type studied in~\cite{svin-sok83}.
WE algebras played 
an important role in obtaining these results about 
BTs in~\cite{cfa,igon-mpi,cfg-2009}. 
A method to obtain results similar to Proposition~\ref{knprop} 
is discussed in Subsection~\ref{pr-backl} of the present paper.

In our opinion, the above-mentioned applications of WE algebras 
strongly suggest to study these algebras for more PDEs. 
According to~\cite{ll}, the WE algebra of the Landau-Lifshitz equation 
is isomorphic to the infinite-dimensional Lie algebra 
of certain matrix-valued functions on an algebraic curve of genus~$1$.
One of our goals is to present examples 
of WE algebras related to higher genus curves.  

To this end,
we study a multicomponent generalization of the Landau-Lifshitz equation 
from~\cite{mll,skr-jmp}.  
To describe this PDE, we need some notation.
Let $\fik$ be either $\Com$ or $\mathbb{R}$.
Fix an integer $n\ge 2$. 
For any $n$-dimensional vectors 
$V={(v^1,\dots,v^n)}$ and $W={(w^1,\dots,w^n)}$, set 
$\langle V,W\rangle=\sum_{i=1}^nv^iw^i$.

Let $r_1,\dots,r_n\in\fik$ be such that $r_i\neq r_j$ for all $i\neq j$. 
Denote by $R=\mathrm{diag}\,(r_1,\dots,r_n)$ the diagonal $(n\times n)$-matrix 
with entries $r_i$.
Consider the PDE
\begin{equation}
\label{main}
S_t=\Big(S_{xx}+\frac32\langle S_x,S_x\rangle S\Big)_x+\frac32\langle S,RS\rangle S_x,
\qquad\quad \langle S,S\rangle=1,\qquad\quad
R=\mathrm{diag}\,(r_1,\dots,r_n),
\end{equation} 
where $S=\big(s^1(x,t),\dots,s^n(x,t)\big)$ 
is a column-vector of dimension~$n$, and 
$s^i(x,t)$ take values in $\fik$.

System~\eqref{main} was introduced in~\cite{mll}. 
According to~\cite{mll}, 
for $n=3$ it coincides with the higher symmetry (the commuting flow) 
of third order for the Landau-Lifshitz equation. 
Thus~\eqref{main} can be regarded as an $n$-component generalization of the Landau-Lifshitz equation. 

The paper~\cite{mll} considers also the following algebraic curve 
\begin{equation}
\label{curve}
\la_i^2-\la_j^2=r_j-r_i,\qquad\qquad i,j=1,\dots,n,
\end{equation}
in the space $\fik^n$ with coordinates $\la_1,\dots,\la_n$.
According to~\cite{mll}, this curve is of genus ${1+(n-3)2^{n-2}}$, 
and system~\eqref{main} possesses a ZCR parametrized by points of this curve. 

System~\eqref{main} has an infinite number of symmetries, 
conservation laws~\cite{mll}, and an auto-B\"acklund transformation with a parameter~\cite{ll-backl}. 
Soliton-like solutions of~\eqref{main} can be found in~\cite{ll-backl}. 
In~\cite{skr-jmp} system~\eqref{main}
and its symmetries are constructed by means of the Kostant--Adler scheme.

The results of this paper can be summarized as follows.  

In Section~\ref{wea} some general properties of WE algebras are presented. 
In particular, a rigorous definition of these algebras is given for a wide class of PDEs. 
An outline of these properties is presented in Subsection~\ref{detdesc}.

In Sections~\ref{comput},~\ref{sect_expl_str},
for all $n\ge 3$, the WE algebra of system~\eqref{main} is computed. 
We prove that the WE algebra of~\eqref{main} is isomorphic 
to the direct sum $\fik^2\oplus L(n)$.
Here $\fik^2$ is a $2$-dimensional abelian Lie algebra,  
and $L(n)$ is an infinite-dimensional Lie algebra 
of certain matrix-valued functions on the curve~\eqref{curve}. 
Applications of this result to some classification problems for 
B\"acklund transformations of~\eqref{main} are discussed in 
Subsection~\ref{pr-backl}. 

To our knowledge, 
this is the first example of a computation of WE algebras for PDEs related to algebraic curves of genus $>1$. 
Also, this seems to be the first example of an explicit description of the WE algebra for a PDE with more than $3$ dependent variables. 
(In system~\eqref{main}, the dependent variables are  $s^1(x,t),\dots,s^n(x,t)$.) 

In Remark~\ref{we-curve} in Subsection~\ref{detdesc} 
we discuss how one can recover the curve~\eqref{curve} 
from the WE algebra of~\eqref{main}.

As a by-product, we obtain a presentation for the algebra~$L(n)$ 
in terms of a finite number of generators and relations. 

The algebra $L(n)$ is very similar to infinite-dimensional 
Lie algebras that were studied in a different context  in~\cite{mll,naukma,skr,skr-jmp}. 
Note that a presentation in terms of a finite number of generators 
and relations was not known for $L(n)$ in the case~$n>3$. 
For~$n=3$ such a presentation was obtained in~\cite{ll}
in the computation of the WE algebra 
of the classical Landau-Lifshitz equation.

In Section~\ref{miura} we construct new 
B\"acklund transformations of Miura type, 
which connect system~\eqref{main} 
with a family of integrable evolution PDEs 
parametrized by the curve~\eqref{curve}.
Also, some solutions of these PDEs are described.
The constructed BTs correspond to certain 
vector field representations of the WE algebra 
of~\eqref{main}. 

These results are explained in more detail in Subsection~\ref{detdesc}.

Weaker versions of some of these results appeared in our preprint~\cite{gll-2008}. 
For completeness, we include some results of~\cite{gll-2008} in the present paper. 

\subsection{A more detailed description of the results}
\label{detdesc}

In Section~\ref{wea} we give a definition of 
WE algebras for evolution systems   
\begin{equation}
\label{sys_intr}
\frac{\pd u^i}{\pd t}
=F^i(u^1,\dots,u^m,\,u^1_1,\dots,u^m_1,\dots,u^1_d,\dots,u^m_d),\,\quad
u^i=u^i(x,t),\,\quad u^i_k=\frac{\pd^k u^i}{\pd x^k},\,\quad 
i=1,\dots,m. 
\end{equation}
The main idea of our definition 
is very similar to that of~\cite{dodd,nonl,pirani79,Prol}. 
However, 
instead of the standard approach of differential forms and vector fields, 
we use formal power series with coefficients in Lie algebras. 
The formal power series approach has the following advantage. 

In the classical Wahlquist-Estabrook 
prolongation theory~\cite{dodd,nonl,pirani79,Prol},
one imposes some conditions on the functions~$F^i$ in~\eqref{sys_intr}, 
in order to get a well-defined WE algebra. 
We do not impose any conditions on~$F^i$. 
The formal power series approach makes it possible 
to define the WE algebra   
rigorously for every system~\eqref{sys_intr}, where $F^i$ can be arbitrary. 

The definition goes as follows. 
Suppose that $u^i$ take values in $\fik$.
Let $D_x$, $D_t$ be the total derivative operators corresponding to~\eqref{sys_intr}. 

Fix $a^i_k\in\fik$ for $i=1,\dots,m$ and $k=0,1,2,\dots$ 
such that the functions $F^i$ from~\eqref{sys_intr} are defined 
on a neighborhood of the point $u^i_k=a^i_k$. 
Here $u^i_0$ is $u^i$. 
Consider the equation 
\begin{equation}\label{zcr_intr}
D_x(B)-D_t(A)+[A,B]=0, 
\end{equation}
where $A$ is a power series in the variables 
$u^i-a^i_0$, and $B$ is a power series in the variables 
$u^i_k-a^i_k$ for $k\le d-1$.  
Here $d\ge 1$ is such that $F^i$ may depend only on $u^j_{l}$ for $l\le d$.

The coefficients of the power series $A$, $B$ are regarded 
as generators of the WE algebra, and 
equation~\eqref{zcr_intr}
provides relations for these generators. 
A more detailed description of this construction 
is given in Section~\ref{defwea}.

Thus the WE algebra is determined by system~\eqref{sys_intr} 
and numbers~$a^i_k$.  
In Section~\ref{we-ai} we show that in many cases the WE  
algebra does not depend on the choice of~$a^i_k$.

\begin{remark}
Let $q$ be a nonnegative integer. 
One can also study equation~\eqref{zcr_intr} in the case when 
$A$ may depend on $u^i_k-a^i_k$ for $k\le q$ and 
$B$ may depend on $u^{i'}_{k'}-a^{i'}_{k'}$ for $k'\le q+d-1$.

If $q=0$, we get the WE algebra.
When $q>0$, the problem becomes much more complicated, 
because one needs to use gauge transformations, in order to simplify solutions $A$, $B$ of~\eqref{zcr_intr}. 
Studying equation~\eqref{zcr_intr} for $q>0$ and using gauge transformations, 
one can obtain the so-called \emph{fundamental Lie algebra} for~\eqref{sys_intr}, which generalizes the WE algebra.
The notion of fundamental Lie algebras for PDEs is described in~\cite{cfg-2009} 
and is briefly discussed in Subsection~\ref{pr-backl} of the present paper.
\end{remark}

In Section~\ref{comput} 
this construction is applied to system~\eqref{main}.
If $n=2$ then~\eqref{main} is equivalent to a scalar equation 
of the form $u_t=u_{xxx}+f(u,u_x,u_{xx})$. 
For scalar equations of this type, WE algebras have already been 
studied quite well (see, e.g.,~\cite{kdv1,cfa} and references therein).  
In the case $n=2$ the curve~\eqref{curve} is rational. 
For these reasons, we assume $n\ge 3$.

Using the definition of WE algebras, 
we first obtain the WE algebra of~\eqref{main} 
in terms of generators and relations. 
Namely, in Section~\ref{comput} 
it is shown that the WE algebra of~\eqref{main}  
is isomorphic to the direct sum of a $2$-dimensional abelian Lie algebra 
and an infinite-dimensional Lie algebra $\mg(n)$. 
The algebra $\mg(n)$ is given by generators $p_1,\dots,p_n$ 
and the relations 
\begin{gather}
\label{int_rel1}
[p_i,[p_j,p_k]]=0,\quad\qquad i\neq j\neq k\neq i,
\quad\qquad i,j,k=1,\dots,n,
\\
\label{int_rel2}
[p_i,[p_i,p_k]]-[p_j,[p_j,p_k]] = (r_j-r_i)p_k,
\qquad i\neq k,\qquad j\neq k,\qquad i,j,k=1,\dots,n.
\end{gather}

In Section~\ref{sect_expl_str} we prove that $\mg(n)$ is isomorphic 
to the infinite-dimensional Lie algebra $L(n)$ 
of certain $\mathfrak{so}_{n,1}$-valued functions on the curve~\eqref{curve}. 
Here $\mathfrak{so}_{n,1}$ is the Lie algebra of the matrix Lie group $\mathrm{O}(n,1)$, 
which consists of linear transformations that preserve the standard 
bilinear form of signature $(n,1)$. 
From the isomorphism $\mg(n)\cong L(n)$  
we get for~$L(n)$ a presentation 
in terms of $n$~generators and relations~\eqref{int_rel1},~\eqref{int_rel2}.

One has also $L(n)=\bigoplus_{i=1}^\infty L_i$ for some vector subspaces 
$L_i\subset L(n)$ with the following properties 
$$
[L_i,L_j]\subset L_{i+j}+L_{i+j-2},
\qquad\quad\dim L_{2k-1}=n,\qquad\quad\dim L_{2k}=\frac{n(n-1)}{2},\qquad\quad 
i,j,k\in\zsp.
$$
Thus the Lie algebra $L(n)$ is quasigraded (almost graded) 
in the sense of~\cite{quasigraded,skr-jmp}. 

For $n=3$ relations~\eqref{int_rel1},~\eqref{int_rel2} and 
the isomorphism $\mg(3)\cong L(3)$ 
were obtained in~\cite{ll} in the computation of the WE algebra 
of the classical Landau-Lifshitz equation.
  
\begin{remark}\label{we-curve}
Clearly, relations~\eqref{int_rel2} look somewhat similar to equations~\eqref{curve}. 
And indeed, formulas~\eqref{qie}, \eqref{is} and Theorem~\ref{gnL} 
in Section~\ref{sect_expl_str} explain how $p_i$ is related to~$\la_i$. 
Relations~\eqref{int_rel2} are obtained 
by the Wahlquist-Estabrook method applied to~\eqref{main}. 
Therefore, at least in some examples, WE algebras help to answer 
the following question. 
Given a PDE, which is suspected to be integrable, 
how to find an algebraic curve such that the PDE 
possesses a ZCR parametrized by this curve?  
More precisely, we mean the following.

According to Section~\ref{defwea}, 
one has a universal procedure that constructs the WE algebra 
in terms of generators and relations for any system~\eqref{sys_intr}. 
Applying this procedure to system~\eqref{main},
one gets relations~\eqref{int_rel1},~\eqref{int_rel2}. 
If we want to find a ZCR parametrized by an algebraic curve, 
we should assume that $p_i$ corresponds to a matrix-valued 
function on a curve. Then, looking at relations~\eqref{int_rel2}, 
one can guess that one should consider the curve~\eqref{curve}. 
\end{remark}

Our proof of the isomorphism $\mg(n)\cong L(n)$ goes as follows. 
The ZCR for~\eqref{main} 
described in~\cite{mll,skr-jmp} can be interpreted 
as a ZCR with values in~$L(n)$. 
This ZCR corresponds to a representation of the WE algebra of~\eqref{main}. 
Therefore, we obtain a homomorphism from the WE algebra to $L(n)$.
Using some filtrations on the algebras $\mg(n)$ and $L(n)$, 
we prove that this homomorphism induces an isomorphism between 
$\mg(n)$ and $L(n)$. 

Recall that a \emph{Miura type transformation} (MTT)
for system~\eqref{sys_intr} is given by  
\begin{gather}
\label{vt}
v^i_t=G^i(v^j,v^j_{x},v^j_{xx},\dots),\quad\qquad v^i=v^i(x,t),
\quad\qquad i,j=1,\dots,m,\\
\label{mt}
u^i=H^i(v^j,v^j_{x},v^j_{xx},\dots),\qquad\qquad i,j=1,\dots,m.
\end{gather} 
Here~\eqref{vt} is another evolution PDE, 
and formulas~\eqref{mt} must satisfy the following properties. 
For any solution $v^i$ of~\eqref{vt}, the functions $u^i$ given 
by~\eqref{mt} obey equations~\eqref{sys_intr}. 
And for any solution $u^i$ of~\eqref{sys_intr}, 
locally there exist functions $v^i$ 
satisfying~\eqref{vt},~\eqref{mt}. 
 
MTTs play an essential role in the classification 
of some types of integrable PDEs (see, e.g.,~\cite{svin-sok83}). 

To our knowledge, before the present paper, 
there were no examples of MTTs for system~\eqref{main}. 
In Section~\ref{miura} we construct a family of such MTTs parametrized by 
points of the curve~\eqref{curve}. 

Namely, we find an evolution system of the form 
\begin{equation}\label{mtsysmain}
\mtf^i_t=P^i(\la_1,\dots,\la_n,
\mtf^j,\mtf^j_x,\mtf^j_{xx},\mtf^j_{xxx}),\quad\qquad i=1,\dots,n,\quad\qquad 
\sum_{i=1}^n(\mtf^i)^2=1,
\end{equation}
and a transformation 
\begin{equation}\label{siwi}
s^i=R^i(\la_1,\dots,\la_n,\mtf^j,\mtf^j_x),\qquad\qquad i=1,\dots,n.
\end{equation}
Here $\la_1,\dots,\la_n\in\Com$ are parameters satisfying~\eqref{curve} 
and $\la_i\neq 0$ for all $i=1,\dots,n$. 
Formulas~\eqref{mtsysmain},~\eqref{siwi} are defined locally 
on some open subset of the space of jets of functions $s^i$, $\mtf^i$. 

For any solution $\mtf^1,\dots,\mtf^n$ of~\eqref{mtsysmain}, the function 
$S=(s^1,\dots,s^n)$ given by~\eqref{siwi} obeys~\eqref{main}. 
For any fixed solution $S=(s^1,\dots,s^n)$ of~\eqref{main} 
and any fixed nonzero numbers $\la_1,\dots,\la_n$ satisfying~\eqref{curve}, 
locally there is an $(n-1)$-parametric family of solutions 
$\mtf^1,\dots,\mtf^n$ of equations~\eqref{mtsysmain},~\eqref{siwi}.

This seems to be the first example of MTTs parametrized 
by an algebraic curve of genus $>1$.
To construct this MTT, 
we find a nonlinear reduction of the auxiliary linear system 
corresponding to the ZCR for~\eqref{main}. 

It is well known that, if system~\eqref{sys_intr} is integrable 
and system~\eqref{vt} is connected with~\eqref{sys_intr} by an MTT~\eqref{mt}, 
then~\eqref{vt} is also integrable.   
Therefore, since~\eqref{main} is integrable, 
we see that~\eqref{mtsysmain} is integrable as well. 
In particular, one can transfer the known ZCR, conservation laws, 
and auto-B\"acklund transformations of~\eqref{main} to system~\eqref{mtsysmain} 
by means of the transformation~\eqref{siwi}. 

In Remark~\ref{mttvector} it is shown that the constructed MTTs 
correspond to some vector field representations of the WE 
algebra of~\eqref{main}.
In Section~\ref{miura}
we show also how to obtain solutions for~\eqref{mtsysmain} 
from solutions of~\eqref{main} and describe some solutions for~\eqref{mtsysmain} explicitly. 

Section~\ref{sec_ap} contains the proof of the technical Lemma~\ref{lemma} 
about $\mg(n)$.

\begin{remark}
Several more integrable PDEs with ZCRs parametrized by the 
curve~\eqref{curve} were introduced in~\cite{mll,naukma,skr}. 
It was noticed in~\cite{skr} that the formulas
$\la=\la_i^2+r_i,\,\ y=\prod_{i=1}^n\la_i$ 
provide a map from the curve~\eqref{curve}
to the hyperelliptic curve $y^2=\prod_{i=1}^n(\la-r_i)$.
According to~\cite{mll}, for $n>3$ 
the curve~\eqref{curve} itself is not hyperelliptic.
\end{remark}

\subsection{Some problems on B\"acklund transformations}
\label{pr-backl}

In this subsection, all functions are assumed to be analytic. 
Recall that system~\eqref{main} is determined by 
constants $r_1,\dots,r_n$.
Denote system~\eqref{main} by $\mathbf{L}(r_1,\dots,r_n)$.

Similarly to Proposition~\ref{knprop}, it is natural to ask the following 
question. Let $r_1,\dots,r_n,r'_1,\dots,r'_n\in\fik$ be such that 
$r_i\neq r_j$ and $r'_i\neq r'_j$ for all $i\neq j$. 
Is the system~$\mathbf{L}(r_1,\dots,r_n)$ connected 
with the system~$\mathbf{L}(r'_1,\dots,r'_n)$ 
by any B\"acklund transformation (BT)? 

In other words, we are interested in classification of systems $\mathbf{L}(r_1,\dots,r_n)$ 
for $r_1,\dots,r_n\in\fik$ with respect to B\"acklund transformations.
In the present subsection, 
we would like to discuss some work in progress about questions 
of this type. 

It is well known that, in order to study BTs for a  
PDE~\eqref{sys_intr}, one needs to consider overdetermined systems 
\begin{gather}
\label{pseud}
w^j_x=f^j(w^l,x,t,u^i,u^i_x,u^i_{xx},\dots),\quad\qquad
w^j_t=g^j(w^l,x,t,u^i,u^i_x,u^i_{xx},\dots),\\
\notag 
w^j=w^j(x,t),\qquad\quad  
j,l=1,\dots,q,
\end{gather}
such that system~\eqref{pseud} is compatible modulo~\eqref{sys_intr}.

The WE algebra of~\eqref{sys_intr} helps 
to describe systems of the following type
\begin{equation}\label{wepsp}
w^j_x=f^j(w^l,u^i),\qquad 
w^j_t=g^j(w^l,u^i,u^i_x,u^i_{xx},\dots),\qquad
w^j=w^j(x,t),\qquad 
j,l=1,\dots,q,
\end{equation}
where equations~\eqref{wepsp} 
are assumed to be compatible modulo~\eqref{sys_intr}.
It is well known that systems~\eqref{wepsp} 
correspond to representations of the WE algebra by vector fields 
on the manifold $W$ with coordinates $w^1,\dots,w^q$.

A similar description of systems~\eqref{pseud} is given in~\cite{cfg-2009}. 
We do not have a possibility to report here all details of this theory, 
so we present only a sketchy overview of the main ideas. 

For a given PDE~\eqref{sys_intr}, 
the preprint~\cite{cfg-2009} defines the \emph{fundamental Lie algebra}, 
which generalizes the WE algebra and satisfies the following property. 
Any compatible system~\eqref{pseud} is gauge equivalent 
to a system arising from a vector field representation 
of the fundamental Lie algebra of~\eqref{sys_intr}. 
More precisely, the fundamental Lie algebra is defined 
for each point of the infinite prolongation of~\eqref{sys_intr} 
in the corresponding jet space (see~\cite{cfg-2009} for details).  

This Lie algebra is called fundamental, 
because it is analogous to the fundamental group in topology. 
According to~\cite{nonl}, 
there is a notion of coverings of PDEs such that 
compatible systems~\eqref{pseud} are coverings of~\eqref{sys_intr}.  
This notion is similar to the classical concept 
of coverings from topology. 
Recall that the fundamental group of a manifold~$M$ is responsible 
for topological coverings of~$M$. 
In a somewhat similar way, 
the fundamental Lie algebra of~\eqref{sys_intr} 
is responsible for coverings~\eqref{pseud} of~\eqref{sys_intr}. 
The fundamental Lie algebra of a PDE has also some 
coordinate-independent geometric meaning (see~\cite{cfg-2009}). 

Let $\mathfrak{L}_1$ and $\mathfrak{L}_2$ be Lie algebras. 
We say that $\mathfrak{L}_1$ is 
\emph{cofinitely-equivalent} to $\mathfrak{L}_2$ 
if for each $i=1,2$ there is a subalgebra 
$\mathfrak{H}_i\subset\mathfrak{L}_i$ of finite codimension such that 
$\mathfrak{H}_1$ is isomorphic to $\mathfrak{H}_2$.

For example, let $\mathfrak{L}_1$ be an infinite-dimensional Lie algebra 
and $\mathfrak{L}_2\subset\mathfrak{L}_1$ be a subalgebra of finite codimension. 
Then $\mathfrak{L}_1$ is cofinitely-equivalent to $\mathfrak{L}_2$, 
because one can take $\mathfrak{H}_1=\mathfrak{H}_2=\mathfrak{L}_2$.

The following result is proved in~\cite{cfg-2009}.
\begin{proposition}[\cite{cfg-2009}]
\label{faexbt} 
Let $\CE_1$ and $\CE_2$ be evolution PDEs. 
Suppose that $\CE_1$ and $\CE_2$ are connected by a BT. 
Then for each $i=1,2$ 
there is a point $a_i$ in the infinite prolongation of $\CE_i$ 
such that the fundamental Lie algebra of~$\CE_1$ at the point~$a_1$ 
is cofinitely-equivalent to the fundamental Lie algebra of~$\CE_2$ at~$a_2$. 
\end{proposition}
In fact the preprint~\cite{cfg-2009} proves a more general result 
about PDEs that are not necessarily evolution. 
A result similar to Proposition~\ref{faexbt} is used in~\cite{igon-mpi} 
in order to prove Proposition~\ref{knprop}.

For a given evolution PDE~\eqref{sys_intr}, 
there is a natural homomorphism from the fundamental Lie algebra 
to the WE algebra. This homomorphism reflects the fact that 
systems~\eqref{wepsp} are a particular case of systems~\eqref{pseud}.

Recall that~\eqref{main} is an evolution PDE, 
so we can consider the fundamental Lie algebras of~\eqref{main}. 
These algebras are studied in~\cite{hwemll}.
Fix a point $a$ in the infinite prolongation of~\eqref{main}. 
Denote by $\psi$ the homomorphism from the fundamental Lie algebra 
of~\eqref{main} at~$a$ to the WE algebra of~\eqref{main}. 

As has been said in Subsection~\ref{motiv}, 
the WE algebra is isomorphic to $\fik^2\oplus L(n)$. 
Using this description of the WE algebra, 
the preprint~\cite{hwemll} shows that the image of~$\psi$ 
is isomorphic to $L(n)$. 
The kernel of $\psi$ is studied in~\cite{hwemll} as well.  
Loosely speaking, the results of~\cite{hwemll} imply that 
the ``main part'' of the fundamental Lie algebra of~\eqref{main} 
is equal to the image of~$\psi$ and, therefore, is isomorphic to $L(n)$. 

Thus the structure of the WE algebra (described in the present paper) 
plays a very important 
role in the description of the fundamental Lie algebras for~\eqref{main} 
given in~\cite{hwemll}.  

Also, WE algebras help to obtain a partial answer to the above 
question about $\mathbf{L}(r_1,\dots,r_n)$ and 
$\mathbf{L}(r'_1,\dots,r'_n)$. 
Namely, using Proposition~\ref{faexbt} and 
the results of~\cite{cfg-2009,hwemll}, one can prove the following. 
\begin{statement}\label{statem}
If the WE algebra of $\mathbf{L}(r_1,\dots,r_n)$ 
is not cofinitely-equivalent to the WE algebra of $\mathbf{L}(r'_1,\dots,r'_n)$, 
then $\mathbf{L}(r_1,\dots,r_n)$ is not connected 
with~$\mathbf{L}(r'_1,\dots,r'_n)$ by any BT.
\end{statement}
We do not prove Statement~\ref{statem} in the present paper. 
A proof of this statement will appear elsewhere.

Since we have an explicit description of the WE algebra for~\eqref{main}, 
Statement~\ref{statem} provides an algebraic necessary condition 
for existence of a BT connecting 
$\mathbf{L}(r_1,\dots,r_n)$ and $\mathbf{L}(r'_1,\dots,r'_n)$. 

Recall that the WE algebra of $\mathbf{L}(r_1,\dots,r_n)$ is isomorphic 
to $\fik^2\oplus L(n)$, where $L(n)$ consists of   
certain matrix-valued functions on the curve~\eqref{curve}. 
Similarly to Proposition~\ref{knprop}, 
it is natural to expect that the condition of Statement~\ref{statem}
can be reformulated in terms of properties of algebraic curves 
or other algebraic varieties, but this is not clear yet. 

Note that the present paper is self-contained 
and can be studied independently of~\cite{cfa,igon-mpi,cfg-2009,hwemll}. 

\subsection{Abbreviations and notation}

The following abbreviations and notation are used in the paper. 
WE = Wahlquist-Estabrook, ZCR = zero-curvature representation, 
BT = B\"acklund transformation, MTT = Miura type transformation. 
The symbols $\zsp$ and $\zp$ denote the sets of positive and nonnegative 
integers respectively.

\section{The definition and some properties of Wahlquist-Estabrook algebras} 
\label{wea}

\subsection{The definition of Wahlquist-Estabrook (WE) algebras}
\label{defwea}

The main idea of our definition of WE algebras 
is very similar to that of~\cite{dodd}.
However, 
instead of the standard approach of differential forms and vector fields, 
we use formal power series with coefficients in Lie algebras. 

This will allow us to define 
the WE algebras for any evolution system of the form
\begin{equation}
\label{sys}
\frac{\pd u^i}{\pd t}
=F^i(u^1,\dots,u^m,\,u^1_1,\dots,u^m_1,\dots,u^1_d,\dots,u^m_d),\quad
u^i=u^i(x,t),\quad 
u^i_k=\frac{\pd^k u^i}{\pd x^k},\quad  
i=1,\dots,m. 
\end{equation}
Here the number $d\in\zsp$ is such that $F^i$ may depend only 
on $u^j_k$ for $k\le d$.

Following the jet bundle approach to PDEs~\cite{dodd}, 
we regard 
\begin{equation}
\label{coor}
u^i_k,\quad\qquad i=1,\dots,m,\quad\qquad k\in\zp,\quad\qquad u^i_0=u^i,
\end{equation} 
as coordinates of an infinite-dimensional manifold~$\CE$.

Let $\fik$ be either $\Com$ or $\mathbb{R}$. 
In this paper, 
all vector spaces and algebras are over the field~$\fik$.
The coordinates~$u^i_k$ take values in~$\fik$. 
If $\fik=\Com$ then any function of the variables~$u^i_k$ 
is assumed to be analytic. 
In the case~$\fik=\mathbb{R}$, any function is smooth.

For each $l\in\zp$, consider the manifold $\CE_l\cong\fik^{m(l+1)}$ 
with the coordinates~$u^i_k$ for $k\le l$. 
We have the natural projection $\pi_l\cl\CE\to\CE_l$ that ``forgets'' 
the coordinates $u^{i'}_{k'}$ for $k'>l$. 

The topology on $\CE$ is defined as follows. 
For any $l$ and any open subset $V\subset\CE_l$, 
the preimage~$\pi_l^{-1}(V)\subset\CE$ is, by definition, open in $\CE$. Such subsets form a base of the topology on~$\CE$. 
In other words, we consider the smallest topology on~$\CE$ such that 
all the maps $\pi_l$ are continuous. 

A function $f(u^i_k)$ is called \emph{admissible} 
if $f$ depends only on a finite number of the coordinates~\eqref{coor}.  
Let $\ess$ be an open subset of $\CE$ such that 
the functions~$F^i$ from~\eqref{sys} are defined on~$\ess$.
Denote by $\fanl$ the algebra of $\fik$-valued admissible functions on~$\ess$. 

The \emph{total derivative operators} corresponding to~\eqref{sys} are  
\begin{equation}
\label{tdo}
D_x=\frac{\pd}{\pd x}+
\sum_{i,k}
u^i_{k+1}\frac{\pd}{\pd u^i_k},\qquad\qquad 
D_t=\frac{\pd}{\pd t}+
\sum_{i,k}D_x^{k}(F^i)\frac{\pd}{\pd u^i_k}.
\end{equation}
We regard $D_x$, $D_t$ as derivations of the algebra $\fanl$. 
It is well known that $[D_x,D_t]=0$.

Let $\mcl$ be a Lie algebra. 
An \emph{admissible function with values in $\mcl$} 
is an element of the tensor product $\mcl\otimes_\fik\fanl$. 
From now on, all functions are supposed to be admissible.
One has the Lie bracket on $\mcl\otimes_\fik\fanl$ defined as follows 
${[h_1\otimes f_1,\,h_2\otimes f_2]}={[h_1,h_2]\otimes f_1f_2}$ 
for $h_1,h_2\in\mcl$ and $f_1,f_2\in\fanl$.

Recall that a 
\emph{zero-curvature representation} (ZCR) for system~\eqref{sys} is given by 
a pair of functions~$M$,~$N$ with values in a Lie algebra such that 
\begin{equation}
\label{zcrmn}
D_x(N)-D_t(M)+[M,N]=0.
\end{equation}

In the classical Wahlquist-Estabrook prolongation theory~\cite{dodd},
one imposes some conditions on the functions~$F^i$,~$M$,~$N$. 
These conditions imply that 
$M$ may depend only on~$u^1_0,\dots,u^m_0$ and 
$N$ may depend on $u^i_k$ for $k\le d-1$.

We do not impose any conditions on~$F^i$. 
We simply assume that $M=M(u^i_0)$ may depend only on~$u^1_0,\dots,u^m_0$, 
while $N=N(u^i_k)$ 
can be a function of any finite number of the variables~\eqref{coor}.

According to the next lemma, 
our assumption implies that actually $N(u^i_k)$ 
may depend only on $u^i_k$ for $k\le d-1$. 
This lemma is very similar to well-known computations in 
Wahlquist-Estabrook theory~\cite{dodd}.

\begin{lemma}
\label{nukd1}
If $M=M(u^i_0)$ and $N=N(u^i_k)$ satisfy~\eqref{zcrmn}, then  
$\dfrac{\pd N}{\pd u^j_l}=0$ for all $l\ge d$ 
and $j=1,\dots,m$. 
\end{lemma}
\begin{proof}
Let $s$ be the maximal integer such that 
$\dfrac{\pd N}{\pd u^j_{s}}\neq 0$ for some $j$. 
Suppose ${s\ge d}$. 

Since $F^i$ from~\eqref{sys} do not depend on $u^{i'}_{k'}$ for $k'>d$, 
using formulas~\eqref{tdo}, we obtain 
$$
\frac{\pd}{\pd u^j_{s+1}}\Big(D_x(N)\Big)=\frac{\pd N}{\pd u^j_{s}},
\qquad\qquad 
\frac{\pd}{\pd u^j_{s+1}}\Big(D_t(M)\Big)=
\frac{\pd}{\pd u^j_{s+1}}\Big([M,N]\Big)=0.
$$
Hence, differentiating~\eqref{zcrmn} 
with respect to~$u^j_{s+1}$, 
one gets $\dfrac{\pd N}{\pd u^j_{s}}=0$, which contradicts to our assumption.
\end{proof}

A point of the manifold $\CE$ is determined by the values of 
the coordinates~\eqref{coor} at this point. 
Let $a^i_k\in\fik$ be such that the point 
\begin{equation}
\label{upointmn}
a=\big(u^i_k=a^i_k,\,\ i=1,\dots,m,\,\ k\in\zp\big)\,\in\,\CE
\end{equation}
belongs to $\ess\subset\CE$. 
  
\begin{remark}  
\label{inform}
The main idea of the definition of WE algebras 
can be informally outlined as follows. 
Consider a ZCR of the form $M=M(u^i_0)$, $N=N(u^i_k)$. 
Let $\tilde M$ and $\tilde N$ be the Taylor series of $M$ and $N$ 
at the point~\eqref{upointmn}.
Then $\tilde M$ is a power series in the variables~$u^i_0-a^i_0$, 
and $\tilde N$ is a power series in the variables~$u^i_k-a^i_k$ 
for~$k\le d-1$. 

We regard the coefficients of the power series~$\tilde M,\,\tilde N$ 
as generators of a Lie algebra, and equation~\eqref{zcrmn}
provides relations for these generators. 
As a result, one obtains a Lie algebra given by generators and relations, 
which is called the~\emph{WE algebra} of system~\eqref{sys} at 
the point~\eqref{upointmn}. 
The details of this construction are presented below.

As we will show in Section~\ref{we-ai}, 
in many cases the WE algebra does not depend 
on the choice of numbers~$a^i_k$.
\end{remark}

For each $q\in\zp$, let $\mat_q$ be the set of matrices of 
size~$m\times (q+1)$ with nonnegative integer entries. 
For a matrix $\gamma\in\mat_q$, 
its entries are denoted by $\gamma_{ik}\in\zp$, 
where $i=1,\dots,m$ and $k=0,\dots,q$. 
Let $U^\gamma$ be the following product 
\begin{equation}
\label{ugamma}
U^\gamma=\prod_{\substack{i=1,\dots,m,\\k=0,\dots,q}}
\big(u^i_k-a^i_k\big)^{\gamma_{ik}}.
\end{equation}

We are going to study some formal power series in  
the variables~${u^i_k-a^i_k}$ for $k\le q$. 
Any such series can be written as 
$$
\sum_{\gamma\in \mat_{q}}c_\gamma\cdot U^\gamma,
$$
where $c_\gamma$ are the coefficients of it.
In what follows, 
we will sometimes omit the multiplication sign $\cdot$ in such formulas. 

Let $\frl$ be the free Lie algebra generated 
by the symbols $\fla_\al,\,\flb_\be$ for 
${\al\in\mat_0}$, ${\be\in\mat_{d-1}}$.
Then 
$$
\fla_\al\in\frl,\quad\qquad\flb_\be\in\frl,\qquad\quad
[\fla_\al,\flb_\be]\in\frl\qquad\quad 
\forall\,\al\in\mat_0,\qquad\quad\forall\,\be\in\mat_{d-1}.
$$
Consider the following power series with coefficients in~$\frl$
\begin{equation}
\label{flaflb}
\fla=\sum_{\al\in \mat_0}\fla_\al\cdot U^\al,\qquad\qquad 
\flb=\sum_{\beta\in \mat_{d-1}}\flb_\beta\cdot U^\beta.
\end{equation}

For any $\al\in \mat_0$ and $\beta\in \mat_{d-1}$, 
the expressions $D_x(U^\beta),\,D_t(U^\al)$ 
are functions of a finite number of the variables~$u^i_k$. 
Taking the corresponding Taylor series at the point~\eqref{upointmn}, 
we regard these expressions as power series. 
Let 
\begin{gather}
\label{dxbdta}
D_x(\flb)=\sum_{\beta\in \mat_{d-1}}\flb_\beta\cdot D_x(U^\beta),
\qquad\qquad
D_t(\fla)=\sum_{\al\in \mat_0}\fla_\al\cdot D_t(U^\al),\\
\label{lieab}
[\fla,\flb]=\sum_{\al\in\mat_0,\,\,\beta\in \mat_{d-1}} 
[\fla_\al,\flb_{\beta}]\cdot U^\al\cdot U^{\beta}.
\end{gather}
It is easily seen that 
$D_x(\flb)$, $D_t(\fla)$, $[\fla,\flb]$ 
can be regarded as power series with coefficients in~$\frl$. 
We have 
\begin{equation}
\label{zgamma}
D_x(\flb)-D_t(\fla)+[\fla,\flb]=
\sum_{\gamma\in \mat_{d}}z_\gamma\cdot U^\gamma
\end{equation}
for some $z_\gamma\in\frl$. 
Let $\frid\subset\frl$ be the ideal generated by the 
elements~$z_\gamma$ for all $\gamma\in\mat_{d}$.

The \emph{WE algebra} of system~\eqref{sys} at the point~\eqref{upointmn} 
is defined to be the quotient Lie algebra~$\frl/\frid$.
For $a\in\ess$, the WE algebra at $a$ is denoted by $\wea(a)$.

Let $\mcl$ be a Lie algebra. 
A \emph{formal ZCR at the point~\eqref{upointmn} 
with coefficients in~$\mcl$} is given by power series 
\begin{equation}
\label{hatmn}
\cle=\sum_{\al\in \mat_0}\cle_\al\cdot U^\al,\qquad\qquad 
\dle=\sum_{\beta\in \mat_{d-1}}\dle_\beta\cdot U^\beta,\qquad\qquad
\cle_\al,\dle_\beta\in\mcl,
\end{equation}
such that ${D_x(\dle)-D_t(\cle)+[\cle,\dle]}=0$, 
where $D_x(\dle)$, $D_t(\cle)$, $[\cle,\dle]$ are defined 
similarly to~\eqref{dxbdta}, \eqref{lieab}. 

Consider the natural map 
$\rho\cl\frl\to\frl/\frid=\wea(a)$ and set 
${\wfla_\al=\rho(\fla_\al)}$, ${\wflb_\beta=\rho(\flb_\beta)}$.
The definition of~$\frid$ implies that the power series 
\begin{equation}
\label{zcrwea}
\wfla=\sum_{\al\in \mat_0}\wfla_\al\cdot U^\al,\qquad\qquad 
\wflb=\sum_{\beta\in \mat_{d-1}}\wflb_\beta\cdot U^\beta
\end{equation}
satisfy 
$D_x\big(\wflb\big)-D_t\big(\wfla\big)+\big[\wfla,\wflb\big]=0$. Thus 
$\wfla$, $\wflb$ constitute a formal ZCR with coefficients 
in~$\wea(a)$.
\begin{proposition}
\label{tlhom}
Any formal ZCR~\eqref{hatmn} with coefficients in a Lie algebra~$\mcl$ 
determines a homomorphism $\wea(a)\to\mcl$ given by 
${\wfla_\al\mapsto\cle_\al}$ and ${\wflb_\beta\mapsto\dle_\beta}$. 
\end{proposition}
\begin{proof}
Since $\frl$ is a free Lie algebra generated by 
$\fla_\al$, $\flb_\beta$, 
one can consider the homomorphism $\mu\cl\frl\to\mcl$ 
given by $\mu(\fla_\al)=\cle_\al$, $\mu(\flb_\beta)=\dle_\beta$. 
For any power series of the form 
$$
C=\sum_{\gamma\in \mat_{q}}c_\gamma\cdot U^\gamma,
\quad\qquad c_\gamma\in\frl,\quad\qquad q\in\zp,  
$$
set
$$
\tilde\mu(C)=\sum_{\gamma\in \mat_{q}}\mu(c_\gamma)\cdot U^\gamma.
$$ 
Taking into account~\eqref{flaflb},~\eqref{hatmn}, we get 
\begin{gather}
\notag
\tilde\mu(\fla)=\sum_{\al\in \mat_0}\cle_\al\cdot U^\al=\cle,\qquad\qquad 
\tilde\mu(\flb)=\sum_{\beta\in \mat_{d-1}}\dle_\beta\cdot U^\beta=\dle,\\
\label{zcrzcr0}
\tilde\mu\Big(D_x(\flb)-D_t(\fla)+[\fla,\flb]\Big)=
D_x(\dle)-D_t(\cle)+[\cle,\dle]=0.
\end{gather}
From~\eqref{zgamma},~\eqref{zcrzcr0} it follows that $\mu(z_\gamma)=0$ 
for all $\gamma\in\mat_{d}$ and, therefore, $\mu(\frid)=0$. 
Hence $\mu\cl\frl\to\mcl$ induces the homomorphism 
$\wea(a)=\frl/\frid\to\mcl$ such that $\wfla_\al\mapsto\cle_\al$, 
$\wflb_\beta\mapsto\dle_\beta$.
\end{proof}

\begin{remark}\label{remzcrhw}
Suppose that functions $M(u^i_0)$, $N(u^i_k)$ with values in a Lie algebra~$\mathfrak{L}$ form a ZCR. 
Then the Taylor series of $M(u^i_0)$, $N(u^i_k)$ at the point~\eqref{upointmn} 
constitute a formal ZCR with coefficients in~$\mathfrak{L}$. 
Therefore, by Proposition~\ref{tlhom}, we obtain a homomorphism $\wea(a)\to\mathfrak{L}$.
\end{remark}

\subsection{WE algebras at different points}
\label{we-ai}

\begin{remark}
According to~\cite{dodd,kdv1,schief,ll,Prol} and references therein,
for many PDEs (including the KdV, nonlinear Schr\"odinger, Landau-Lifshitz, Harry-Dym equations), 
the WE algebra does not depend on the choice 
of numbers~$a^i_k$ in~\eqref{upointmn},~\eqref{ugamma}. 

In the present subsection we explain this property. 
The main idea can be outlined as follows. 
For each of these PDEs, 
there is a finite collection of analytic functions 
$$
f_1(u^i_0),\ \quad f_2(u^i_0),\ \quad\dots,
\ \quad f_{\nmf}(u^i_0),\qquad\qquad 
g_1(u^i_k),\ \quad g_2(u^i_k),\ \quad\dots,\ \quad g_{\nmg}(u^i_k)
$$
such that any ZCR $M(u^i_0)$, $N(u^i_k)$ is of the form 
\begin{equation}
\notag
M(u^i_0)=\sum_{j=1}^{\nmf}M_j\cdot f_j(u^i_0),\qquad\qquad N(u^i_k)=\sum_{l=1}^{\nmg}N_l\cdot g_l(u^i_k),
\end{equation} 
where $M_j$, $N_l$ are elements of a Lie algebra and 
satisfy some Lie algebraic relations. 
Using the Taylor series of $M(u^i_0)$ and $N(u^i_k)$ 
at the point~\eqref{upointmn}, we will show that  
$M_j$,~$N_l$ generate the WE algebra.

Since the Lie algebraic relations for $M_j$, $N_l$ do not 
depend on the choice of numbers~$a^i_k$ in~\eqref{upointmn},~\eqref{ugamma}, 
one obtains that the WE algebra does not depend on~$a^i_k$ for such PDEs.
The details of these arguments are presented below.
\end{remark}

We need first some auxiliary constructions. 
Recall that $\fanl$ is the algebra of $\fik$-valued admissible functions on~$\ess$. 
Fix positive integers $\nmf$, $\nmg$ and functions 
$f_j,g_l\in\fanl$ for $j=1,\dots,\nmf$, $l=1,\dots,\nmg$.  

Let $\ufl$ be the free Lie algebra generated by the symbols
$\flm_1,\dots,\flm_{\nmf},\,\fln_1,\dots,\fln_{\nmg}$.
Consider the following element of $\ufl\otimes_\fik\fanl$
$$
\flz=
\sum_{l}\fln_l\otimes D_x(g_l)-\sum_{j}\flm_j\otimes D_t(f_j)
+\sum_{j,l}[\flm_j,\fln_l]\otimes f_jg_l.
$$
An ideal $I\subset\ufl$ is said to be \emph{$\flz$-tame} if 
the natural map ${\ufl\otimes\fanl\to(\ufl/I)\otimes\fanl}$ 
sends~$\flz$ to zero.
Let $\idz\subset\ufl$ be the intersection of all $\flz$-tame ideals of $\ufl$.

\begin{remark}
\label{gridz}
A set of generators for the ideal~$\idz$ 
is constructed as follows. 
Let $v_1,\dots,v_q$ be a basis for the linear subspace of~$\fanl$ 
spanned by the functions $D_x(g_l)$, $D_t(f_j)$, $f_jg_l$. 
Then there are $e_1,\dots,e_q\in\ufl$ such that 
$$
\flz=\sum_{l}\fln_l\otimes D_x(g_l)-\sum_{j}\flm_j\otimes D_t(f_j)
+\sum_{j,l}[\flm_j,\fln_l]\otimes f_jg_l=\sum_{s=1}^qe_s\otimes v_s.
$$
Then $e_1,\dots,e_q$ generate the ideal $\idz$. 
Indeed, since $v_1,\dots,v_q$ are linearly independent, 
the elements $e_1,\dots,e_q$ belong to any $\flz$-tame ideal 
and, therefore, belong to $\idz$. On the other hand, 
the ideal generated by $e_1,\dots,e_q$ is $\flz$-tame and, consequently, 
contains $\idz$. 
In particular, one obtains that the ideal $\idz$ is $\flz$-tame.
\end{remark}

Consider the natural homomorphism 
$\sigma\cl\ufl\to\ufl/\idz$ and set 
${\wflm_j=\sigma(\flm_j)}$, ${\wfln_l=\sigma(\fln_l)}$.
Since $\idz$ is $\flz$-tame, one has
\begin{equation}
\label{zcrwmn}
\sum_{l}\wfln_l\otimes D_x(g_l)-\sum_{j}\wflm_j\otimes D_t(f_j)
+\sum_{j,l}[\wflm_j,\wfln_l]\otimes f_jg_l=0.
\end{equation}

\begin{theorem}
\label{wepoints}
Consider an evolution system~\eqref{sys} and the corresponding 
manifold $\CE$ with coordinates~\eqref{coor}.  
Suppose that there are a connected open subset $\ess\subset\CE$ 
and analytic functions 
$$
f_1(u^i_0),\ \quad f_2(u^i_0),\ \quad\dots,
\ \quad f_{\nmf}(u^i_0),\qquad\qquad 
g_1(u^i_k),\ \quad g_2(u^i_k),\ \quad\dots,\ \quad g_{\nmg}(u^i_k)
$$
on $\ess$ such that the following properties hold. 
\begin{itemize}
\item
The functions~$F^i$ from~\eqref{sys} are analytic on $\ess$. 

\item  
For any point~\eqref{upointmn} of $\ess$, any Lie algebra~$\mcl$, 
and any formal ZCR~\eqref{hatmn}, one has 
\begin{equation}
\label{cdpqfg}
\cle=\sum_{j=1}^{\nmf} \ple_j\cdot f_j,\qquad\qquad 
\dle=\sum_{l=1}^{\nmg} \qle_l\cdot g_l
\end{equation}
for some elements $\ple_j,\qle_l\in\mcl$.
In formulas~\eqref{cdpqfg} we regard $f_j$,~$g_l$ as power series, 
using the Taylor series of~$f_j$,~$g_l$ at the point~\eqref{upointmn}.

\end{itemize}
Consider the algebra~$\ufl$ and the ideal $\idz\subset\ufl$ 
corresponding to $f_1,\dots,f_{\nmf},\,g_1,\dots,g_{\nmg}$,  
as constructed above. 

Then for any $a\in\ess$ the WE algebra $\wea(a)$ of system~\eqref{sys} 
is isomorphic to $\ufl/\idz$. 
Hence for any $a,a'\in\ess$ one has $\wea(a)\cong\wea(a')$.
\end{theorem}
\begin{proof}
Recall that \eqref{zcrwea} is a formal ZCR with coefficients in~$\wea(a)$. 
Applying the assumption of the theorem to this formal ZCR, we get 
\begin{equation}
\label{weamn}
\wfla=\sum_{j=1}^{\nmf}\ple_j\cdot f_j,\qquad\qquad 
\wflb=\sum_{l=1}^{\nmg}\qle_l\cdot g_l
\end{equation}
for some elements $\ple_j,\,\qle_l\in\wea(a)$. 

Since $\ess$ is connected, any analytic function on $\ess$ 
is uniquely determined by its Taylor series at the point~\eqref{upointmn}. 
Therefore, the identity  ${D_x\big(\wflb\big)-D_t\big(\wfla\big)+\big[\wfla,\wflb\big]}=0$ 
is equivalent to the following equation in the space $\wea(a)\otimes_\fik\fanl$ 
\begin{equation}
\label{pqfgwea}
\sum_{l}\qle_l\otimes D_x(g_l)-\sum_{j}\ple_j\otimes D_t(f_j)
+\sum_{j,l}[\ple_j,\qle_l]\otimes f_jg_l=0.
\end{equation}
Since $\ufl$ is a free Lie algebra generated by 
$\flm_j$, $\fln_l$, one can consider the homomorphism 
$\tau\cl\ufl\to\wea(a)$ given by $\tau(\flm_j)=\ple_j$, $\tau(\fln_l)=\qle_l$. 

Equation~\eqref{pqfgwea} says that 
the kernel of $\tau$ is $\flz$-tame and, therefore, ${\tau(\idz)=0}$.
Thus $\tau\cl\ufl\to\wea(a)$ induces the homomorphism 
$\varphi\cl\ufl/\idz\to\wea(a)$ such that $\varphi\big(\wflm_j\big)=\ple_j$ and 
$\varphi\big(\wfln_l\big)=\qle_l$.

Using the Taylor series of $f_j$ and $g_l$, we regard the expressions   
$$
\wflm=\sum_{j}\wflm_j\cdot f_j,\qquad\qquad\wfln=\sum_{l}\wfln_l\cdot g_l,
$$
as power series with coefficients in~$\ufl/\idz$. 

Equation~\eqref{zcrwmn} implies  ${D_x\big(\wfln\big)-D_t\big(\wflm\big)+\big[\wflm,\wfln\big]=0}$. 
Hence $\wflm$, $\wfln$ constitute a formal ZCR. 

Let $\psi\cl\wea(a)\to\ufl/\idz$ be the homomorphism 
corresponding to this formal ZCR by Proposition~\ref{tlhom}.

It is easy to check that the constructed homomorphisms 
$\varphi\cl\ufl/\idz\to\wea(a)$ and $\psi\cl\wea(a)\to\ufl/\idz$ 
are inverse to each other. 
Thus for any $a\in\ess$ the algebra~$\wea(a)$ is 
isomorphic to $\ufl/\idz$. 
Hence
for any $a,a'\in\ess$ one has $\wea(a)\cong\ufl/\idz\cong\wea(a')$. 
\end{proof}
\begin{remark}
\label{remgenwe}
According to Section~\ref{defwea}, 
in general the WE algebra is given by an infinite number 
of generators and relations. 
However, if the assumptions of Theorem~\ref{wepoints} are satisfied, 
then the WE algebra is isomorphic to~$\ufl/\idz$, which is given 
by a finite number of generators and relations. 

Indeed, the elements 
$\wflm_1,\dots,\wflm_{\nmf},\wfln_1,\dots,\wfln_{\nmg}$ generate $\ufl/\idz$. 
Relations for these generators are given 
by $e_1,\dots,e_q\in\idz$ constructed in Remark~\ref{gridz}. 
\end{remark}

\begin{example}
To clarify the constructions of this section, 
consider a simple example in the case $m=1$. 
Set $u=u^1$. 
Let us describe generators and relations for 
the WE algebra of the equation $u_t=u_{xx}$. 

Similarly to~\eqref{coor}, 
we regard $u_k={\pd^k u}/{\pd x^k}$ 
as coordinates of the corresponding manifold~$\CE$.  
Formulas~\eqref{tdo} become 
$D_x={\pd_x+\sum_{k\ge 0}u_{k+1}\pd_{u_k}}$ and 
$D_t={\pd_t+\sum_{k\ge 0}u_{k+2}\pd_{u_k}}$, where $u_0=u$.

For a Lie algebra~$\mcl$, 
a formal ZCR~\eqref{hatmn} at a point $a\in\CE$ is given by formal power series 
\begin{gather}
\notag
\cle=\sum_{i_0=0}^\infty\cle_{i_0} (u_0-a_0)^{i_0},\qquad\quad 
\dle=\sum_{i_0,i_1=0}^\infty
\dle_{i_0,i_1}(u_0-a_0)^{i_0}(u_1-a_1)^{i_1},\qquad\quad 
\cle_{i_0},\dle_{i_0,i_1}\in\mcl,\\
\label{zcutu2}
D_x(\dle)-D_t(\cle)+[\cle,\dle]=0,
\end{gather}
where the numbers $a_k\in\fik$ determine the point~$a\in\CE$, 
similarly to~\eqref{upointmn}.

It is easy to check that equation~\eqref{zcutu2} is satisfied 
if and only if $\cle$, $\dle$ are of the form   
\begin{equation}
\label{cdpqfgu2}
\cle=\ple_1\cdot u_0+\ple_2\cdot 1,\qquad\qquad
\dle=\qle_1\cdot u_1+\qle_2\cdot u_0+\qle_3\cdot 1
\end{equation}
for some $\ple_j,\qle_l\in\mcl$ satisfying 
$\qle_1-\ple_1=0$, $[\ple_1,\qle_1]=0$, $\qle_2+[\ple_2,\qle_1]=0$, 
$[\ple_1,\qle_2]=0$, $[\ple_1,\qle_3]+[\ple_2,\qle_2]=0$, $[\ple_2,\qle_3]=0$.
According to formulas~\eqref{cdpqfgu2}, 
one can apply Theorem~\ref{wepoints} for 
$\ess=\CE$, $\nmf=2$, $\nmg=3$, $f_1=u_0$, $f_2=1$, $g_1=u_1$, $g_2=u_0$, $g_3=1$.

By Theorem~\ref{wepoints}, 
the WE algebra at any point $a\in\CE$ is isomorphic to~$\ufl/\idz$. 
Applying Remark~\ref{remgenwe} to this example, 
we obtain that the algebra $\ufl/\idz$ is given by the generators 
$\wflm_1$, $\wflm_2$, $\wfln_1$, $\wfln_2$, $\wfln_3$ and the relations
$\wfln_1-\wflm_1=0$, $[\wflm_1,\wfln_1]=0$, $\wfln_2+[\wflm_2,\wfln_1]=0$,
$[\wflm_1,\wfln_2]=0$, $[\wflm_1,\wfln_3]+[\wflm_2,\wfln_2]=0$, $[\wflm_2,\wfln_3]=0$.
\end{example}
According to the computations of~\cite{dodd,schief,ll,Prol} 
and references therein,
Theorem~\ref{wepoints} is applicable also to 
the KdV, nonlinear Schr\"odinger, Landau-Lifshitz, 
Harry-Dym equations, and many other analytic evolution PDEs. 
Although the papers~\cite{dodd,schief,ll,Prol} 
consider only smooth or analytic ZCRs, 
for these PDEs the computations essentially remain the same 
for any formal ZCRs~\eqref{hatmn}, 
so one can apply Theorem~\ref{wepoints}. 
In Section~\ref{comput} we will show that Theorem~\ref{wepoints} 
is applicable also to system~\eqref{main}, 
if we rewrite this system as~\eqref{sp},~\eqref{pt}.

\section{The WE algebra of the multicomponent Landau-Lifshitz system}
\label{comput}

For any $m\in\zp$ and $m$-dimensional vectors 
$v={(v^1,\dots,v^m)}$, $w={(w^1,\dots,w^m)}$, set 
$\langle v,w\rangle=\sum_{i=1}^mv^iw^i$.

In order to compute the WE algebra of system~\eqref{main},
we need to resolve the constraint $\langle S,S\rangle=1$
for the vector-function $S=\big(s^1(x,t),\dots,s^n(x,t)\big)$.
Following~\cite{mll}, we do this as 
\begin{equation}
\label{sp}
s^j=\frac{2u^j}{1+\langle u,u\rangle},\qquad\qquad
j=1,\dots,n-1,\qquad\qquad
s^n=\frac{1-\langle u,u\rangle}{1+\langle u,u\rangle},
\end{equation}
where $u=\big(u^1(x,t),\dots,u^{n-1}(x,t)\big)$ 
is an $(n-1)$-dimensional vector-function.

We assume $n\ge 3$. 
The reasons for this assumption were explained in Section~\ref{detdesc}.

As is shown in~\cite{mll}, using~\eqref{sp}, 
one can rewrite system~\eqref{main} as 
\begin{multline}
\label{pt}
u_t=u_{xxx}-6\langle u,u_x\rangle\Delta^{-1}u_{xx}
+\bigl(-6\langle u,u_{xx}\rangle\Delta^{-1}+24\langle u,u_{x}\rangle^2\Delta^{-2}
-6\langle u,u\rangle\langle u_x,u_{x}\rangle\Delta^{-2}\bigl)u_x+\\
+\bigl(6\langle u_x,u_{xx}\rangle\Delta^{-1}-12\langle u,u_x\rangle\langle u_x,u_{x}\rangle\Delta^{-2}\bigl)u
+\frac32\Bigl(r_n+4\Delta^{-2}\sum_{i=1}^{n-1}(r_i-r_n)(u^i)^2\Bigl)u_x,
\end{multline}
where $\Delta=1+\langle u,u\rangle$, and $r_1,\dots,r_n$ 
are the distinct numbers such that 
$R=\mathrm{diag}\,(r_1,\dots,r_n)$ in~\eqref{main}.

Set $u^i_k=\dfrac{\pd^k u^i}{\pd x^k}$ for $i=1,\dots,n-1$ and $k\in\zp$. 
In particular, $u^i_0=u^i$.
Similarly to~\eqref{coor}, we regard $u^i_k$ as coordinates of 
the corresponding manifold~$\CE$. 
Recall that $u^i_k$ take values in~$\fik$, 
where $\fik$ is either $\Com$ or $\mathbb{R}$. 
For simplicity of notation, we will write $u^i$ instead of~$u^i_0$.

Since the right hand-side of~\eqref{pt} contains negative powers 
of~$\Delta=1+\sum_{i}(u^i)^2$, we introduce the following open subset~$\ess\subset\CE$ 
$$
\ess=\Big\{\big(u^1,\dots,u^{n-1},
u^1_1,\dots,u^{n-1}_1,u^1_2,\dots,u^{n-1}_2,\dots\big)\in\CE
\,\ \Big|\,\ 1+\sum_i(u^i)^2\neq 0\Big\}.
$$  
System~\eqref{pt} is of the form
\begin{equation}\label{eq.landau}
\frac{\pd u^j}{\pd t}=u^j_3+G^j\big(u^i,u^i_1,u^i_2\big),\qquad\qquad 
j=1,\dots,n-1,
\end{equation}
and the functions $G^j\big(u^i,u^i_1,u^i_2\big)$ are analytic on~$\ess$. 

According to Section~\ref{wea}, 
in order to compute the WE algebra of~\eqref{pt}, 
we need to study the equation 
\begin{gather}
\label{gc}
D_x(B)-D_t(A)+[A,B]=0,\\
\notag
A=A(u^1,\dots,u^{n-1}),\qquad\quad
B=B(u^1,\dots,u^{n-1},u^1_1,\dots,u^{n-1}_1,u^1_2,\dots,u^{n-1}_2).
\end{gather}
Here $A$, $B$ can be either smooth functions 
with values in a Lie algebra~$\mcl$
or formal power series with coefficients in~$\mcl$.

In the case of smooth functions, we assume that $A$, $B$ are defined 
on a connected open subset of~$\ess$. 

In the case of formal power series, one has
$$
A=\sum_{i_1,\dots,i_{n-1}\ge 0} A_{i_1\dots i_{n-1}}
(u^1-a^1_0)^{i_1}\dots(u^{n-1}-a^{n-1}_0)^{i_{n-1}},\qquad\quad
A_{i_1\dots i_{n-1}}\in\mcl,
$$
and $B$ is a power series in the variables 
${u^i-a^i_0}$, ${u^i_1-a^i_1}$, ${u^i_2-a^i_2}$ for some fixed numbers $a^i_k\in\fik$ satisfying ${1+\sum_{i=1}^{n-1}(a^i_0)^2\neq 0}$.

We will show that in both cases equation~\eqref{gc} implies that $A$, $B$
are of the form
\begin{equation}
\notag
A=\sum_{j=1}^{\nmf} \ple_j\cdot f_j(u^i),\qquad\quad 
B=\sum_{l=1}^{\nmg} 
\qle_l\cdot g_l(u^i,u^i_1,u^i_2),\qquad\quad\ple_j,\qle_l\in\mcl,
\end{equation}
for some functions $f_j(u^i)$, $g_l(u^i,u^i_1,u^i_2)$, which
are certain polynomials in $s^m$, $D_x(s^m)$, $D_x^2(s^m)$.
Here $s^m=s^m(u^i)$ for ${m=1,\dots,n}$ are given by~\eqref{sp}.
In particular, the functions $f_j$, $g_l$
will be analytic on~$\ess$, 
so we will be able to use Theorem~\ref{wepoints}.

Differentiating equation~\eqref{gc} with respect to~$u^i_{3}$, 
we see that $B$ is of the form 
\begin{equation}\label{baf}
B=\sum_{i=1}^{n-1} u^i_{2} A_{u^i}+\fnt,
\end{equation} 
where $\fnt$ may depend only on $u^j$ and $u^j_1$. 
Here and below, 
the subscripts~$u^i$ denote derivatives with respect to~$u^i$.
That is, ${A_{u^i}=\pd A/\pd u^i}$. 

Then equation~\eqref{gc} becomes 
\begin{equation}\label{eq.covering.2}
\sum_{i,j=1}^{n-1}u^i_2 u^j_{1} A_{u^i u^j}+
\sum_{j=1}^{n-1}\Big(u^j_{1} \fnt_{u^j}+u^j_{2}\frac{\pd\fnt}{\pd u^j_1}
-G^jA_{u^j}+u^j_{2}\left[A,A_{u^j}\right]\Big)+[A,\fnt]=0,
\end{equation}
where $G^j$ is defined by~\eqref{pt},~\eqref{eq.landau} and satisfies 
\begin{equation}
\label{glui2}
\frac{\pd G^j}{\pd u^i_{2}}=\Delta^{-1}
\Big(-6\delta_{ij}\sum_{k} u^ku^k_1-6u^iu^j_1+6u^i_1u^j\Big)
\qquad\qquad\forall\,j,i.
\end{equation}

Differentiating~\eqref{eq.covering.2} with respect to~$u^i_{2}$ 
and using~\eqref{glui2}, one gets
\begin{equation}
\label{fntu1}
\frac{\pd\fnt}{\pd u^i_1}=
-\sum_{j=1}^{n-1}\Big(u^j_1A_{u^iu^j}+
\Delta^{-1}\Big(6\delta_{ij}\sum_{k} u^ku^k_1 
+6u^iu^j_1-6u^i_1u^j\Big)A_{u^j}\Big)-[A,A_{u^i}]\quad\qquad\forall\,i.
\end{equation} 

Integrating equations~\eqref{fntu1} with respect to~$u^i_1$, 
we obtain that $\fnt$ is of the form
\begin{equation}\label{eq.F1.WE}
\fnt=-\frac12\sum_{i,j}u_1^{i}u_1^{j}
\big(A_{u^{i}u^{j}}+12u^{i}\Delta^{-1}A_{u^{j}}\big)
+\sum_{i,j}3\Delta^{-1}\big(u_{1}^{i}\big)^{2}u^jA_{u^j}
+\sum_{i} u_{1}^{i}[A_{u^i},A] + \hlt,
\end{equation}
where $\hlt$ may depend only on $u^1,u^2,\dots,u^{n-1}$.

Substituting~\eqref{eq.F1.WE} in~\eqref{eq.covering.2}, 
we see that the left-hand side of~\eqref{eq.covering.2} 
is a third degree polynomial in~$u^i_1$. 
Equating to zero the coefficients of~$u^{i_1}_1u^{i_2}_1u^{i_3}_1$ 
of this polynomial, one gets the following equations 
\begin{gather}
\label{auuu1}
\begin{split}
A_{u^{i}u^{i}u^{i}}&=6\Delta^{-1}
\bigg(\sum_ku^{k}A_{u^{i}u^{k}}-2u^{i}A_{u^{i}u^{i}}
-A_{u^{i}}\bigg)+\\
&+12\Delta^{-2}\bigg(
\sum_ku^{i}u^{k}A_{u^{k}}+\langle u,u\rangle A_{u^{i}}
-2\big(u^{i}\big)^{2}A_{u^{i}}\bigg)\quad\qquad\forall\,i, 
\end{split}\\
\label{auuu2}
\begin{split}
A_{u^{i}u^{i}u^{h}}&=2\Delta^{-1}
\bigg(\sum_ku^{k}A_{u^{h}u^{k}}-4u^{i}A_{u^{i}u^{h}}
-A_{u^{h}}-2u^{h}A_{u^{i}u^{i}}\bigg)+\\
&+4\Delta^{-2}\bigg(\sum_ku^hu^k A_{u^k}-4u^{i}u^{h}A_{u^{i}} -2\big(u^{i}\big)^{2}A_{u^{h}}+\langle u,u\rangle A_{u^{h}}\bigg)\quad\qquad 
\forall\,i\neq h, 
\end{split}\\
\label{auuu3}
\begin{split}
A_{u^{i}u^{j}u^{h}}&=-4\Delta^{-1}
\big(u^{j}A_{u^{i}u^{h}}+u^{i}A_{u^{j}u^{h}}
+u^{h}A_{u^{i}u^{j}}\big)+\\
&+8\Delta^{-2}
\left(-u^{j}u^{h}A_{u^{i}}-
u^{i}u^{h}A_{u^{j}}-u^{i}u^{j}A_{u^{h}}\right)
\quad\qquad\forall\,i<j<h.
\end{split}
\end{gather}

\begin{proposition}\label{accks} 
Let $A=A(u^1,\dots,u^{n-1})$ be either a smooth function 
with values in a Lie algebra~$\mcl$
or a formal power series with coefficients in~$\mcl$.
Then $A$ satisfies~\eqref{auuu1},~\eqref{auuu2},~\eqref{auuu3} 
if and only if 
\begin{equation}\label{eq.X.nice.bis}
A=C_0+\sum_{l=1}^{n}C_ls^l
\end{equation}
for some $C_0,C_1,\dots,C_n\in\mcl$. 
Here the functions $s^l=s^l(u^1,\dots,u^{n-1})$ are given by~\eqref{sp}.
\end{proposition}

\begin{remark}
\label{sk_solution}
We would like to explain how one can guess that 
$A$ in~\eqref{gc} must be of the form~\eqref{eq.X.nice.bis}.
Since the original system~\eqref{main} is written in terms 
of $S=(s^1,\dots,s^n)$, 
it is natural to expect that $A$ can be expressed in terms of~$s^l$.
Then the simplest possibility is that $A$ depends linearly on~$s^l$.
According to Proposition~\ref{accks}, 
this natural guess turns out to be correct.

For $n=3$ some analog of formula~\eqref{eq.X.nice.bis} 
appears in the description of ZCRs of the classical Landau-Lifshitz equation~\cite{ll}.
\end{remark}

\begin{proof}[Proof of Proposition~\textup{\ref{accks}}]
We regard~\eqref{auuu1},~\eqref{auuu2},~\eqref{auuu3} as PDEs for 
${A=A(u^1,\dots,u^{n-1})}$.
Let us compute some differential consequences of these PDEs.

Denote by~$R(i)$ the right-hand side of~\eqref{auuu1} 
and by~$\tilde{R}(i,h)$ the right-hand side of~\eqref{auuu2}. 
For any $i\neq h$, 
let us differentiate equation~\eqref{auuu1} with respect to~$u^h$ 
and equation~\eqref{auuu2} with respect to~$u^i$. One gets 
\begin{equation}\label{auuuua}
A_{u^{i}u^{i}u^{i}u^{h}}=\frac{\pd}{\pd u^h}\Big(R(i)\Big),\qquad\qquad
A_{u^{i}u^{i}u^{h}u^{i}}=\frac{\pd}{\pd u^i}\Big(\tilde{R}(i,h)\Big).
\end{equation}
Since $A_{u^{i}u^{i}u^{i}u^{h}}=A_{u^{i}u^{i}u^{h}u^{i}}$, 
equations~\eqref{auuuua} imply
\begin{equation}\label{ririh}
\frac{\pd}{\pd u^h}\Big(R(i)\Big)=\frac{\pd}{\pd u^i}\Big(\tilde{R}(i,h)\Big)
\qquad\qquad\forall\,i\neq h.
\end{equation}
Equations~\eqref{ririh} are PDEs of third order for~$A$. 
Let us replace the third order derivatives of~$A$ by the right-hand sides 
of~\eqref{auuu1},~\eqref{auuu2},~\eqref{auuu3}. 
Then equations~\eqref{ririh} become PDEs of second order. 
It is straightforward to show that the obtained system of second order PDEs is equivalent to 
\begin{equation}\label{eq.appp}
A_{u^{i}u^{h}}=-2\Delta^{-1}\big(u^hA_{u^{i}} + u^i A_{u^{h}}\big) 
\qquad\qquad\forall\,i\neq h.
\end{equation}

Since $\tilde{R}(i,h)$ is the right-hand side of~\eqref{auuu2}, 
one has ${A_{u^{i}u^{i}u^{h}}=\tilde{R}(i,h)}$.
Differentiating~\eqref{eq.appp} with respect to~$u^i$ 
and replacing $A_{u^{i}u^{h}u^{i}}$ by~$\tilde{R}(i,h)$, we obtain
\begin{equation}
\label{trihpd}
\tilde{R}(i,h)=\frac{\pd}{\pd u^i}
\Big(-2\Delta^{-1}\big(u^hA_{u^{i}} + u^i A_{u^{h}}\big)\Big)
\qquad\qquad\forall\,i\neq h.
\end{equation}
Using~\eqref{eq.appp}, 
in~\eqref{trihpd} we can replace $A_{u^{j}u^{l}}$ by 
$-2\Delta^{-1}\big(u^lA_{u^{j}} + u^j A_{u^{l}}\big)$ for any $j\neq l$.
As a result, one gets 
\begin{equation}\label{eq.apppp}
\big(A_{u^iu^i} - A_{u^ju^j}\big)
+4\Delta^{-1}\big(u^iA_{u^i} - u^jA_{u^j}\big)=0 
\qquad\qquad\forall\,i\neq j.
\end{equation}

Consider first the case when $A$ 
is a formal power series with coefficients in~$\fik$.
\begin{lemma}
\label{fpsass}
Let $a^1_0,\dots,a^{n-1}_0\in\fik$ be such that 
${1+\sum_i(a^i_0)^2\neq 0}$.
A formal power series 
\begin{equation}
\label{aiaiu}
A=\sum_{i_1,\dots,i_{n-1}\ge 0} A_{i_1\dots i_{n-1}}
(u^1-a^1_0)^{i_1}\dots(u^{n-1}-a^{n-1}_0)^{i_{n-1}},\qquad\qquad
A_{i_1\dots i_{n-1}}\in\fik,
\end{equation}
satisfies~\eqref{auuu1},~\eqref{auuu2},~\eqref{auuu3} iff 
${A=b_0+\sum_{l=1}^{n}b_ls^l}$ for some ${b_0,b_1,\dots,b_n\in\fik}$, 
where $s^l=s^l(u^1,\dots,u^{n-1})$ are given by~\eqref{sp}. 

Here we regard the functions $s^l=s^l(u^1,\dots,u^{n-1})$ 
as power series, using the corresponding Taylor series at the 
point~$u^i=a^i_0$.
\end{lemma}
\begin{proof}
Let $\mathcal{V}$ be the vector space of formal power series~\eqref{aiaiu} 
satisfying~\eqref{auuu1},~\eqref{auuu2},~\eqref{auuu3}. 
If $A\in\mathcal{V}$ then $A$ obeys also~\eqref{eq.appp},~\eqref{eq.apppp}. 
Let $A\in\mathcal{V}$ be given by~\eqref{aiaiu}. 
According to~\eqref{auuu1}, \eqref{auuu2}, \eqref{auuu3}, 
any third order derivative of~$A$ 
is expressed in terms of lower order derivatives.
Therefore, 
if $A_{i_1\dots i_{n-1}}=0$ for all $i_1,\dots,i_{n-1}\ge 0$ such that   ${i_1+\dots+i_{n-1}\le 2}$, then $A=0$.   
 
Combining this with~\eqref{eq.appp},~\eqref{eq.apppp}, we see the following. 
If $A_{20\dots 0}=0$ and $A_{j_1\dots j_{n-1}}=0$ for all $j_1,\dots,j_{n-1}\ge 0$ satisfying ${j_1+\dots+j_{n-1}\le 1}$, then $A=0$. 

Thus any power series $A\in\mathcal{V}$ is uniquely determined 
by the coefficients 
$$
A_{20\dots 0},\qquad A_{j_1\dots j_{n-1}},\qquad\quad 
j_1,\dots,j_{n-1}\ge 0,\qquad j_1+\dots+j_{n-1}\le 1,
$$
hence $\dim\mathcal{V}\le n+1$. 
It is easy to check that the functions 
\begin{equation}
\label{sss1}
1,\qquad s^1(u^1,\dots,u^{n-1}),\qquad s^2(u^1,\dots,u^{n-1}),\qquad \dots,\qquad 
s^n(u^1,\dots,u^{n-1})
\end{equation} 
satisfy PDEs~\eqref{auuu1},~\eqref{auuu2},~\eqref{auuu3}.
The functions~\eqref{sss1} are linearly independent over~$\fik$ 
and are analytic on a neighborhood of the point~$u^i=a^i_0$. 
Therefore, the Taylor series of the functions~\eqref{sss1} 
are linearly independent and belong to~$\mathcal{V}$. 
Since $\dim\mathcal{V}\le n+1$, this implies the statement of the lemma.  
\end{proof}
Now let us study the case 
when $A$ is a smooth function with values in~$\fik$. 

\begin{lemma}
\label{smabs}
Consider the space $\fik^{n-1}$ with the coordinates 
${u^1,\dots,u^{n-1}}$ and the open subset 
$$
U=\Big\{\big(u^1,\dots,u^{n-1}\big)\in\fik^{n-1}
\,\ \Big|\,\ 1+\sum_i(u^i)^2\neq 0\,\Big\}.
$$ 
Let $W$ be a connected open subset of~$U$. 
A smooth $\fik$-valued function $A(u^1,\dots,u^{n-1})$ on~$W$ 
satisfies \eqref{auuu1}, \eqref{auuu2}, \eqref{auuu3} iff 
${A=b_0+\sum_{l=1}^{n}b_ls^l}$ for some ${b_0,b_1,\dots,b_n\in\fik}$.
\end{lemma}
\begin{proof}
If $\fik=\mathbb{C}$, 
then, according to the assumptions of Section~\ref{defwea}, 
the function $A(u^1,\dots,u^{n-1})$ is analytic 
and the statement follows from Lemma~\ref{fpsass}.

Consider the case $\fik=\mathbb{R}$. 
Since the functions~\eqref{sss1} satisfy PDEs~\eqref{auuu1},~\eqref{auuu2},~\eqref{auuu3}, 
we see that $b_0+\sum_{l=1}^{n}b_ls^l$ obeys these PDEs 
for any $b_0,\dots,b_n\in\fik$. 

Suppose that a smooth function $A=A(u^1,\dots,u^{n-1})$ on~$W$ 
satisfies~\eqref{auuu1},~\eqref{auuu2},~\eqref{auuu3}. 

Let $p\in W$. 
Applying Lemma~\ref{fpsass} to the Taylor series of~$A$ at the point $p\in W$, we obtain the following.  
There are $b_0,\dots,b_n\in\fik$ such that the function 
${\tilde{A}=A-b_0-\sum_{l}b_ls^l}$ satisfies $\tilde{A}(p)=0$ and all partial derivatives of~$\tilde{A}$ vanish at~$p$. 
It remains to prove that $\tilde{A}(p')=0$ for any $p'\in W$.

Since $W$ is connected, there is a smooth map $\vf\cl[0,1]\to W$ 
such that ${\vf(0)=p}$ and ${\vf(1)=p'}$, 
where $[0,1]\subset\mathbb{R}$ is the unit interval.
Set
\begin{equation*}
\psi_0(y)=\tilde{A}\big(\vf(y)\big),
\quad\psi_i(y)=\frac{\pd\tilde{A}}{\pd{u^i}}\big(\vf(y)\big),
\quad i=1,\dots,n-1,\quad
\psi_n(y)=\frac{\pd^2\tilde{A}}{\pd{u^1}\pd{u^1}}\big(\vf(y)\big),
\quad y\in[0,1].
\end{equation*}

Since $A$ satisfies \eqref{auuu1}, \eqref{auuu2}, \eqref{auuu3}, \eqref{eq.appp}, 
\eqref{eq.apppp}, the function~$\tilde{A}$ obeys these PDEs as well.  
According to~\eqref{auuu1}, \eqref{auuu2}, \eqref{auuu3}, 
any third order derivative of~$\tilde{A}$ is expressed linearly 
in terms of lower order derivatives of~$\tilde{A}$.
Equations~\eqref{eq.appp}, \eqref{eq.apppp} say that any second order derivative of~$\tilde{A}$ is expressed linearly in terms of 
$\tilde{A}_{u^1},\dots,\tilde{A}_{u^{n-1}},\,\tilde{A}_{u^1u^1}$.

This implies that $\psi_0,\dots,\psi_n$ satisfy 
some linear ordinary differential equations 
\begin{equation}
\label{psiode}
\frac{d\psi_i}{dy}=\sum_{j=0}^{n}g_{ij}(y)\psi_j(y),
\qquad\qquad i=0,1,\dots,n.
\end{equation}
Since $\psi_j(0)=0$ for all $j=0,1,\dots,n$, 
equations~\eqref{psiode} imply ${\psi_j(1)=0}$. 
Hence $\tilde{A}(p')=\psi_0(1)=0$.
\end{proof}

Return to the proof of Proposition~\ref{accks}.

Consider the case when $A$ is a smooth function with values in~$\mcl$. 
That is, $A$ belongs to the tensor product 
$\mcl\otimes_\fik\fanl_0$, where $\fanl_0$ is the space 
of $\fik$-valued smooth functions in the variables $u^1,\dots,u^{n-1}$. 

There are linearly independent elements $E_1,\dots,E_q\in\mcl$ 
such that $A=\sum_{r=1}^q{E_r\otimes A^r}$ for some $A^r\in\fanl_0$.
Then $A$ satisfies PDEs~\eqref{auuu1}, \eqref{auuu2}, \eqref{auuu3} 
iff for all $r=1,\dots,q$ the function~$A^r$ obeys these PDEs. 
Then formula~\eqref{eq.X.nice.bis} follows 
from Lemma~\ref{smabs} applied to~$A^r$.

Finally, it remains to study the case when $A$ is 
a formal power series
\begin{equation*}
A=\sum_{i_1,\dots,i_{n-1}\ge 0} A_{i_1\dots i_{n-1}}
(u^1-a^1_0)^{i_1}\dots(u^{n-1}-a^{n-1}_0)^{i_{n-1}},\qquad\qquad A_{i_1\dots i_{n-1}}\in\mcl.
\end{equation*}
Denote by $V\subset\mcl$ the vector subspace spanned by 
$A_{j_1\dots j_{n-1}}$ for $j_1+\dots+j_{n-1}\le 2$.
 
Let $D_1,\dots,D_q$ be a basis of~$V$.
Equations~\eqref{auuu1},~\eqref{auuu2},~\eqref{auuu3} imply that 
$A_{i_1\dots i_{n-1}}\in V$ for all ${i_1,\dots,i_{n-1}}$.
Therefore, $A$ obeys~\eqref{auuu1},~\eqref{auuu2},~\eqref{auuu3} iff
$A$ is of the form ${A=\sum_{r=1}^q D_r\tilde{A}^r}$, 
where $\tilde{A}^r$ are power series with coefficients in~$\fik$ 
and satisfy~\eqref{auuu1},~\eqref{auuu2},~\eqref{auuu3}. 
Then formula~\eqref{eq.X.nice.bis} follows 
from Lemma~\ref{fpsass} applied to~$\tilde{A}^r$.
\end{proof}

Recall that the left-hand side of~\eqref{eq.covering.2} 
is a third degree polynomial in~$u^i_1$.
As we have shown above, 
the coefficients of~$u^{i_1}_1u^{i_2}_1u^{i_3}_1$ of this polynomial 
vanish iff $A$ is of the form~\eqref{eq.X.nice.bis}. 
Therefore, 
from now on we can assume that $A$ is given by~\eqref{eq.X.nice.bis}.

Substituting~\eqref{eq.X.nice.bis} in~\eqref{eq.F1.WE} and~\eqref{eq.covering.2}, one obtains that the coefficients of~$u^{i_1}_1u^{i_2}_1$ in~\eqref{eq.covering.2} vanish iff 
\begin{equation}\label{eq.algebra.1.bis}
[C_0,C_k]=0,\qquad\qquad k=1,\dots,n.
\end{equation}

Equating to zero the linear in~$u^j_1$ part of~\eqref{eq.covering.2}, 
we get 
\begin{equation}\label{eq.diff.of.F2}
\hlt_{u^j}=\frac{3}{2}\sum_{i,k=1}^{n}r_iC_k\big(s^i\big)^2 s^k_{u^j}  
- \sum_{i,m,k=1}^{n}[C_i,[C_m,C_k]]s^i s^m_{u^j} s^k.
\end{equation}
Recall that 
the subscripts~$u^j$ denote derivatives with respect to~$u^j$. 
So, $\hlt_{u^j}=\pd\hlt/\pd{u^j}$ and $s^k_{u^j}=\pd s^k/\pd {u^j}$.

Differentiating~\eqref{eq.diff.of.F2} with respect to~$u^h$, one obtains
\begin{equation}
\label{f2ujuh}
\hlt_{u^j u^h}=\frac32
\sum_{i,k=1}^{n}r_iC_k\Big(2s^i s^i_{u^h} s^k_{u^j} + \big(s^i\big)^2 s^k_{u^j u^h}\Big)
- \sum_{i,m,k=1}^{n} [C_i,[C_m,C_k]]\big(s^i_{u^h} s^m_{u^j} s^k + s^i s^m_{u^j u^h} s^k + s^i s^m_{u^j} s^k_{u^h}\big).
\end{equation}
Since $\hlt_{u^j u^h}=\hlt_{u^h u^j}$, equations~\eqref{f2ujuh} imply 
\begin{multline}
\label{sumsss}
\sum_{i,m,k=1}^{n} [C_i,[C_m,C_k]]
\big(s^i_{u^h} s^m_{u^j} s^k - s^i_{u^j} s^m_{u^h} s^k 
+ s^i s^m_{u^j} s^k_{u^h} - s^i s^m_{u^h} s^k_{u^j}\big)=\\
=3\sum_{i,k=1}^{n}r_iC_k
\big(s^i s^i_{u^h} s^k_{u^j} - s^i s^i_{u^j} s^k_{u^h}\big),\qquad
j,h=1,\dots,n-1.
\end{multline}
Substituting~\eqref{sp} in~\eqref{sumsss}, 
we obtain that equations~\eqref{sumsss} are equivalent to 
\begin{gather}
\label{cccijk0}
[C_i,[C_j,C_k]]=0,\quad\qquad i\neq j\neq k\neq i,
\quad\qquad i,j,k=1,\dots,n,
\\
\label{ccciik}
[C_i,[C_i,C_k]]-[C_j,[C_j,C_k]] = (r_j-r_i)C_k,
\qquad i\neq k,\qquad j\neq k,\qquad i,j,k=1,\dots,n.
\end{gather}

Set 
\begin{equation}
\label{lkdefin}
Y_1=[C_2,[C_2,C_1]]+r_2C_1,\qquad\quad
Y_m=[C_1,[C_1,C_m]]+r_1C_m,\qquad 
m=2,3,\dots,n.
\end{equation}
From~\eqref{ccciik},~\eqref{lkdefin} it follows that 
\begin{equation}
\label{ljcici} 
Y_j=[C_i,[C_i,C_j]]+r_iC_j\qquad\qquad
\forall\,i\neq j,\qquad\qquad i,j=1,\dots,n.
\end{equation}

Using~\eqref{cccijk0},~\eqref{ljcici}, and  
$\sum_i(s^i)^2=1$, $\sum_is^is^i_{u^j}=0$, 
we can rewrite~\eqref{eq.diff.of.F2} as 
\begin{equation}
\label{f2ujk1}
\hlt_{u^j}= \sum_{k=1}^{n}Y_ks_{u^{j}}^{k}+ 
\frac{1}{2}\sum_{i,k=1}^{n}r_iC_k\big(s^i\big)^2 s^k_{u^j} + 
\sum_{i,k=1}^{n} r_iC_k s^i s^i_{u^j} s^k,\quad\qquad j=1,\dots, n-1.
\end{equation}
Integrating equations~\eqref{f2ujk1} with respect to~$u^j$,
we see that $\hlt$ is of the form 
\begin{equation}\label{eq.F2.new.bis}
\hlt=\sum_{k=1}^n Y_ks^k+
\frac{1}{2}\sum_{i,k=1}^n r_iC_k s^k\big(s^i\big)^2+C_{n+1}\qquad 
\text{for some}\,\ 
C_{n+1}\in\mcl.
\end{equation}
 
Then equation~\eqref{eq.covering.2} reduces to $[A,\hlt]=0$. 
Using~\eqref{eq.X.nice.bis},~\eqref{eq.algebra.1.bis},~\eqref{eq.F2.new.bis}, 
one shows that the equation~$[A,\hlt]=0$ is equivalent to
\begin{equation}
\label{ccn1m1s}
[C_0,C_{n+1}] + \sum_{l=1}^n s^l[C_l,C_{n+1}] + \sum_{l,k=1}^n s^l s^k[C_l,Y_k]=0.
\end{equation}
To study equation~\eqref{ccn1m1s}, we need the following lemma.
\begin{lemma}
\label{lemmacpyq}
Recall that $n\ge 3$.
If $C_1,\dots,C_n\in\mcl$ satisfy~\eqref{cccijk0},~\eqref{ccciik} then
\begin{equation}
\label{cplq}
[C_p,Y_q]=-[C_q,Y_p],\qquad\qquad p,q=1,\dots,n.
\end{equation}
\end{lemma}
\begin{proof}
Let $l\in\{1,\dots,n\}$ be such that 
$l\neq p$, $l\neq q$. 
By~\eqref{ljcici}, 
\begin{equation}\label{lplql}
Y_p=[C_l,[C_l,C_p]]+r_lC_p,\qquad\qquad
Y_q=[C_l,[C_l,C_q]]+r_lC_q.
\end{equation} 

Consider first the case $p\neq q$. 
Using the Jacobi identity and~\eqref{cccijk0}, we get  
$$
[C_p,[C_l,[C_l,C_q]]]=[[C_p,C_l],[C_l,C_q]]+[C_l,[C_p,[C_l,C_q]]]=
[[C_p,C_l],[C_l,C_q]],
$$
because $[C_p,[C_l,C_q]]=0$ by~\eqref{cccijk0}.
Similarly, one has $[C_q,[C_l,[C_l,C_p]]]=[[C_q,C_l],[C_l,C_p]]$.
Therefore,
\begin{multline*}
[C_p,Y_q]+[C_q,Y_p]=[C_p,[C_l,[C_l,C_q]]+r_l[C_p,C_q]+
[C_q,[C_l,[C_l,C_p]]+r_l[C_q,C_p]=\\
=[C_p,[C_l,[C_l,C_q]]]+[C_q,[C_l,[C_l,C_p]]]
=[[C_p,C_l],[C_l,C_q]]+[[C_q,C_l],[C_l,C_p]]=0.
\end{multline*}
Consider the case $p=q$. By~\eqref{lplql}, 
for $p=q$ equation~\eqref{cplq} is equivalent to 
\begin{equation}\label{cpcllcp}
[C_p,[C_l,[C_l,C_p]]]=0,
\end{equation}
so we need to prove~\eqref{cpcllcp}.
Applying $\mathrm{ad}\,C_k$ to~\eqref{ccciik}, we get
\begin{equation}\label{cccckiik} 
[C_k,[C_i,[C_i,C_k]]]=[C_k,[C_j,[C_j,C_k]]],
\qquad i\neq k,\qquad j\neq k.
\end{equation}
By the Jacobi identity,
\begin{equation}\label{ccjacob} 
[C_k,[C_i,[C_i,C_k]]]=[C_i,[C_k,[C_i,C_k]]]=-[C_i,[C_k,[C_k,C_i]]].
\end{equation}
Let $m\in\{1,\dots,n\}$ be such that $m\neq p$, $m\neq l$. 
Using~\eqref{cccckiik},~\eqref{ccjacob}, one obtains 
\begin{multline*}
[C_p,[C_l,[C_l,C_p]]]=[C_p,[C_m,[C_m,C_p]]]=-[C_m,[C_p,[C_p,C_m]]]=\\
=-[C_m,[C_l,[C_l,C_m]]]=[C_l,[C_m,[C_m,C_l]]]=[C_l,[C_p,[C_p,C_l]]].
\end{multline*}
On the other hand, by~\eqref{ccjacob},  
$[C_p,[C_l,[C_l,C_p]]]=-[C_l,[C_p,[C_p,C_l]]]$.
Therefore, we get~\eqref{cpcllcp}.
\end{proof}

Since $[C_l,Y_k]=-[C_k,Y_l]$ by Lemma~\ref{lemmacpyq}, 
one has $\sum_{l,k=1}^{n} s^l s^k[C_l,Y_k]=0$. 
Hence equation~\eqref{ccn1m1s} reads
\begin{equation}
\label{ccn1m1snew}
[C_0,C_{n+1}] + \sum_{l=1}^n s^l[C_l,C_{n+1}]=0.
\end{equation}
Since~$1,\,s^1,\,s^2,\dots,s^n$ are linearly independent, 
equation~\eqref{ccn1m1snew} is equivalent to
\begin{equation}\label{rel.F2.abelian}
[C_k,C_{n+1}]=0,\qquad\qquad k=0,1,2,\dots,n.
\end{equation}

Combining \eqref{sp}, \eqref{baf}, \eqref{eq.F1.WE}, \eqref{eq.X.nice.bis}, \eqref{lkdefin}, \eqref{eq.F2.new.bis}, 
one obtains
\begin{multline}\label{bfinal}
B=D_x^2(A)+[D_x(A),A]+\frac32\sum_{i,k=1}^nC_ks^k\big(D_x(s^i)\big)^2+
\frac{1}{2}\sum_{i,k=1}^n r_iC_k s^k\big(s^i\big)^2+\\
+\big([C_2,[C_2,C_1]]+r_2C_1\big)s^1
+\sum_{j=2}^{n} \big([C_1,[C_1,C_j]]+r_1C_j\big)s^j+C_{n+1}.
\end{multline}

Thus we get the following result. 
\begin{theorem}
\label{theor-wezcrll}
Suppose that $n\ge 3$. 
Let 
$$
A=A(u^1,\dots,u^{n-1}),\quad\qquad
B=B(u^1,\dots,u^{n-1},u^1_1,\dots,u^{n-1}_1,u^1_2,\dots,u^{n-1}_2)
$$
be either smooth functions with values in a Lie algebra~$\mcl$
or formal power series with coefficients in~$\mcl$.

Then $A$, $B$ satisfy the ZCR equation $D_x(B)-D_t(A)+[A,B]=0$
for system~\eqref{pt} if and only if $A$, $B$ are of the 
form~\eqref{eq.X.nice.bis},~\eqref{bfinal}, where $C_0,C_1,\dots,C_{n+1}\in\mcl$ 
obey~\eqref{eq.algebra.1.bis}, \eqref{cccijk0}, \eqref{ccciik}, \eqref{rel.F2.abelian} and the functions $s^i=s^i(u^1,\dots,u^{n-1})$ are given by~\eqref{sp}.
\end{theorem}

Theorem~\ref{theor-wezcrll} implies that system~\eqref{pt} satisfies 
the conditions of Theorem~\ref{wepoints}.
This allows us to give the following description of the WE algebra 
of~\eqref{pt}.

\begin{theorem}
\label{theor-wea}
Let $n\ge 3$. 
For any point $a\in\ess$, 
the WE algebra $\wea(a)$ of system~\eqref{pt} 
is isomorphic to the Lie algebra given by generators $p_0,p_1,\dots,p_{n+1}$ and the relations 
\begin{gather}
\label{p0pn1}
[p_0,p_l]=[p_{n+1},p_l]=[p_0,p_{n+1}]=0,\quad\qquad l=1,\dots,n,\\
\label{pppijk0}
[p_i,[p_j,p_k]]=0,\quad\qquad i\neq j\neq k\neq i,
\quad\qquad i,j,k=1,\dots,n,
\\
\label{pppiik}
[p_i,[p_i,p_k]]-[p_j,[p_j,p_k]] = (r_j-r_i)p_k,
\qquad i\neq k,\qquad j\neq k,\qquad i,j,k=1,\dots,n.
\end{gather}
The algebra $\wea(a)$ is isomorphic 
to the direct sum $\fik^2\oplus\mg(n)$. Here $\mg(n)$ 
is the subalgebra generated by $p_1,\dots,p_n$, 
and $\fik^2$ is the abelian subalgebra spanned by $p_0$, $p_{n+1}$.
\end{theorem}
\begin{proof}
Let $\mathfrak{H}$ be the Lie algebra given by 
generators $p_0,p_1,\dots,p_{n+1}$ and relations~\eqref{p0pn1},
\eqref{pppijk0}, \eqref{pppiik}. 
From~\eqref{p0pn1}, \eqref{pppijk0}, \eqref{pppiik} 
it follows that $\mathfrak{H}$ is isomorphic to $\fik^2\oplus\mg(n)$, 
where $\mg(n)$ is the subalgebra generated by $p_1,\dots,p_n$, 
and $\fik^2$ is the abelian subalgebra spanned by $p_0$, $p_{n+1}$.  

We are going to construct an isomorphism $\mathfrak{H}\cong\wea(a)$ 
similarly to the proof of Theorem~\ref{wepoints}.
In Section~\ref{defwea}, for any system~\eqref{sys}, 
we defined a formal ZCR with coefficients in the WE algebra of~\eqref{sys}. 
Let $A$, $B$ be the power series with coefficients in~$\wea(a)$ 
that determine this ZCR for system~\eqref{pt}. 

Applying Theorem~\ref{theor-wezcrll} to $\mcl=\wea(a)$, 
we obtain that $A$, $B$ are of the form~\eqref{eq.X.nice.bis},~\eqref{bfinal} 
for some elements $C_0,C_1,\dots,C_{n+1}\in\wea(a)$.
Since $C_0,C_1,\dots,C_{n+1}\in\wea(a)$ 
satisfy~\eqref{eq.algebra.1.bis}, \eqref{cccijk0}, \eqref{ccciik}, \eqref{rel.F2.abelian}, 
one has the homomorphism $\varphi\cl\mathfrak{H}\to\wea(a)$
given by $\varphi(p_i)=C_i$. 

On the other hand, by Theorem~\ref{theor-wezcrll}, the formulas $\tilde{A}=p_0+\sum_{l=1}^{n}p_ls^l$ 
and 
\begin{multline*}
\tilde{B}=
D_x^2\big(\tilde{A}\big)+\big[D_x\big(\tilde{A}\big),\tilde{A}\big]
+\frac32\sum_{i,k=1}^np_ks^k\big(D_x(s^i)\big)^2+
\frac{1}{2}\sum_{i,k=1}^n r_ip_k s^k\big(s^i\big)^2+\\
+\big([p_2,[p_2,p_1]]+r_2p_1\big)s^1
+\sum_{j=2}^{n} \big([p_1,[p_1,p_j]]+r_1p_j\big)s^j+p_{n+1}
\end{multline*}
determine a ZCR with values in $\mathfrak{H}$.  
Applying Proposition~\ref{tlhom} and Remark~\ref{remzcrhw} to this ZCR, 
we get a homomorphism $\psi\cl\wea(a)\to\mathfrak{H}$. 
It is easy to verify that the constructed homomorphisms 
$\varphi\cl\mathfrak{H}\to\wea(a)$ and $\psi\cl\wea(a)\to\mathfrak{H}$ 
are inverse to each other. 
\end{proof}

\begin{remark}\label{remhomwea}
Theorems~\ref{theor-wezcrll},~\ref{theor-wea} imply 
that any ZCR~\eqref{eq.X.nice.bis},~\eqref{bfinal}  
with values in a Lie algebra~$\mcl$ determines a homomorphism 
${\wea(a)\to\mcl}$ given by $p_i\mapsto C_i$.
\end{remark}

\section{The explicit structure of the WE algebra}
\label{sect_expl_str}

Let $\mg(n)$ be the Lie algebra given by generators $p_1,\dots,p_n$ 
and the relations 
\begin{gather}
\label{rel1}
[p_i,[p_j,p_k]]=0,\quad\qquad i\neq j\neq k\neq i,
\quad\qquad i,j,k=1,\dots,n,
\\
\label{rel2}
[p_i,[p_i,p_k]]-[p_j,[p_j,p_k]] = (r_j-r_i)p_k,
\qquad i\neq k,\qquad j\neq k,\qquad i,j,k=1,\dots,n.
\end{gather}

According to Theorem~\ref{theor-wea}, 
the WE algebra of system~\eqref{pt} is isomorphic to $\fik^2\oplus\mg(n)$.
To describe the explicit structure of~$\mg(n)$, 
we need some auxiliary constructions. 

Denote by $\mathfrak{gl}_{n+1}(\fik)$ the space of matrices of 
size $(n+1)\times(n+1)$ with entries from~$\fik$.
Let $E_{i,j}\in\mathfrak{gl}_{n+1}(\fik)$ be the matrix with 
$(i,j)$-th entry equal to 1 and all other entries equal to zero. 

The Lie subalgebra $\mathfrak{so}_{n,1}\subset\mathfrak{gl}_{n+1}(\fik)$ 
was defined in Section~\ref{detdesc}. 
It has the following basis
$$
E_{i,j}-E_{j,i},\qquad i<j\le n,\qquad\quad 
E_{l,n+1}+E_{n+1,l},\qquad l=1,\dots,n.
$$
From the results of~\cite{mll,skr-jmp} one can obtain the following 
$\mathfrak{so}_{n,1}$-valued ZCR for system~\eqref{main}
\begin{gather}
\label{M}
M=\sum_{i=1}^ns^i\la_i(E_{i,n+1}+E_{n+1,i}),\\
\label{N}
N=D_x^2(M)+[D_x(M),M]
+\Big(r_1+\la_1^2+\frac12\langle S,RS\rangle+\frac32\langle S_x,S_x\rangle\Big)M,\\
\notag
D_x(N)-D_t(M)+[M,N]=0.
\end{gather}
Here $\la_1,\dots,\la_n\in\fik$ are parameters satisfying~\eqref{curve}. 
If $S=(s^1,\dots,s^n)$ is given by formulas~\eqref{sp} then~\eqref{M},~\eqref{N} 
determine a ZCR for system~\eqref{pt}. 

Let us regard $\la_1,\dots,\la_n$ as abstract variables 
and consider the algebra $\fik[\la_1,\dots,\la_n]$ 
of polynomials in $\la_1,\dots,\la_n$. 
Let $\ipol\subset\fik[\la_1,\dots,\la_n]$ be the ideal  
generated by $\la_i^2-\la_j^2+r_i-r_j$ for $i,j=1,\dots,n$. 
 
Consider the quotient algebra $\qalg=\fik[\la_1,\dots,\la_n]/\ipol$.
If $\fik=\Com$ then $\qalg$ is isomorphic to the 
algebra of polynomial functions on the algebraic curve~\eqref{curve}. 

The space $\mathfrak{so}_{n,1}\otimes_\fik \qalg$ 
is an infinite-dimensional Lie algebra over $\fik$ 
with the Lie bracket 
$$
[M_1\otimes h_1,\,M_2\otimes h_2]=[M_1,M_2]\otimes h_1h_2,
\qquad\quad 
M_1,M_2\in \mathfrak{so}_{n,1},\qquad\quad h_1,h_2\in \qalg.
$$
We have the natural homomorphism 
$\xi\cl\fik[\la_1,\dots,\la_n]\to\fik[\la_1,\dots,\la_n]/\ipol=\qalg$. 
Set $\hat\la_i=\xi(\la_i)\in \qalg$. 

Formula~\eqref{M} suggests to study 
the following elements of ${\mathfrak{so}_{n,1}\otimes \qalg}$
\begin{equation}
\label{qie}
Q_i=(E_{i,n+1}+E_{n+1,i})\otimes\hat\la_i,
\qquad\qquad i=1,\dots,n.
\end{equation}
Denote by $L(n)\subset\mathfrak{so}_{n,1}\otimes \qalg$ the Lie subalgebra  
generated by~$Q_1,\dots,Q_n$. 

To construct a basis for $L(n)$, 
we need to describe some properties of~$\qalg$. 

Since $\hat{\la}_i^2-\hat\la_j^2+r_i-r_j=0$ in $\qalg$,
the element 
$\hat\la=\hat\la_i^2+r_i\in \qalg$ does not depend on $i$.  
\begin{lemma} 
\label{lemq} 
The elements 
\begin{equation}
\label{elemq}
{\hat\la}^k\hat\la_l,\qquad {\hat\la}^k\hat\la_i\hat\la_j,
\qquad i,j,l\in\{1,\dots,n\},\qquad 
i<j,\qquad k\in\zp,
\end{equation}
are linearly independent over $\fik$.
\end{lemma}
\begin{proof}
Suppose that some linear combination 
of the elements~\eqref{elemq} is zero in $\qalg$ 
\begin{equation}\label{lincomb}
\sum_{l,k}
a_{lk}{\hat\la}^k\hat\la_l+
\sum_{i,j,k,\ i<j}
b_{ijk}{\hat\la}^k\hat\la_i\hat\la_j=0,\quad\qquad a_{lk},b_{ijk}\in\fik,
\end{equation}
where only a finite number of the coefficients 
$a_{lk}$, $b_{ijk}$ may be nonzero.
Set 
$$
\Psi_1=\sum_{l,k} a_{lk}(\la_1^2+r_1)^k\la_l,\qquad\qquad
\Psi_2=\sum_{i,j,k,\ i<j}b_{ijk}(\la_1^2+r_1)^k\la_i\la_j,\qquad\qquad \Psi=\Psi_1+\Psi_2. 
$$
Since $\xi(\la_1^2+r_1)=\hat\la$, 
the left-hand side of~\eqref{lincomb} is equal to $\xi(\Psi)$. 
Hence~\eqref{lincomb} is equivalent to $\Psi\in \ipol$. 

For $l=1,\dots,n$, let $\rho_l$  
be the automorphism of the algebra $\fik[\la_1,\dots,\la_n]$ given by 
$\rho_l(\la_l)=-\la_l$ and $\rho_l(\la_i)=\la_i$ for all $i\neq l$.  
Obviously, $\rho_l(\ipol)=\ipol$.

One has $(\rho_1\rho_2\dots\rho_n)(\Psi)=\Psi-2\Psi_1$. 
Since $\Psi\in \ipol$, we obtain $\Psi_1\in \ipol$ and $\Psi_2=\Psi-\Psi_1\in \ipol$. 
Then the identity $\rho_l(\Psi_1)=\Psi_1-2\la_l\sum_k a_{lk}(\la_1^2+r_1)^k$ 
implies 
\begin{equation}\label{lamini}
\la_l\sum_k a_{lk}(\la_1^2+r_1)^k\in \ipol,\qquad\qquad l=1,\dots,n.
\end{equation}

We have also $\rho_m(\Psi_2)=\Psi_2-2\Phi_m$ for all $m=1,\dots,n$, where
$$
\Phi_m=\la_m\bigg(\sum_{i,k,\ i<m}b_{imk}(\la_1^2+r_1)^k\la_i
+\sum_{j,k,\ j>m}b_{mjk}(\la_1^2+r_1)^k\la_j\bigg). 
$$
Therefore, since $\Psi_2\in \ipol$, one gets $\Phi_m\in \ipol$.
Then the identity 
$$
\rho_i(\Phi_m)=\Phi_m-2\la_m\la_i\sum_kb_{imk}(\la_1^2+r_1)^k\qquad\qquad 
\forall\,i<m,
$$
yields
\begin{equation}
\label{lamlal}
\la_m\la_i\sum_k b_{imk}(\la_1^2+r_1)^k\in \ipol\qquad\qquad \forall\,i<m.
\end{equation}

Suppose that $a_{lk_0}\neq 0$ for some $l\in\{1,\dots,n\}$ and $k_0\in\zp$. 
Then there exists $c_1\in\Com$ such that 
\begin{equation}
\label{cccneq0}
\sum_k a_{lk}(c_1^2+r_1)^k\neq 0,\qquad\qquad
c_1^2+r_1-r_l\neq 0.
\end{equation}
Let $c_2,c_3,\dots,c_n\in\Com$
be such that $c_q^2=c_1^2+r_1-r_q$ for $q=2,3,\dots,n$.
Then $c_i^2-c_j^2+r_i-r_j=0$ for all $i,j=1,\dots,n$. 
Therefore, $P(c_1,\dots,c_n)=0$ 
for any polynomial $P(\la_1,\dots,\la_n)\in \ipol$. 

From~\eqref{cccneq0} we get 
$c_l\sum_k a_{lk}(c_1^2+r_1)^k\neq 0$, which contradicts to~\eqref{lamini}. 
Hence $a_{lk}=0$ for all $l,k$.

Similarly,~\eqref{lamlal} implies $b_{imk}=0$ for all $k$ and $i<m$. 
Thus we have proved that equation~\eqref{lincomb} yields $a_{lk}=b_{ijk}=0$. 
Therefore, the elements~\eqref{elemq} are linearly independent. 
\end{proof}

For $i,j\in\{1,\dots,n\}$ and $k\in\zsp$, 
consider the following elements of 
${\mathfrak{so}_{n,1}\otimes_\fik \qalg}$  
$$
Q^{2k-1}_i=(E_{i,n+1}+E_{n+1,i})\otimes\hat\la^{k-1}\hat\la_i,\qquad\qquad
Q^{2k}_{ij}=(E_{i,j}-E_{j,i})\otimes\hat\la^{k-1}\hat\la_i\hat\la_j. 
$$
For $i,j,l,m\in\{1,\dots,n\}$ and $k_1,k_2\in\zsp$ one has 
\begin{multline} 
\label{q1}
[Q^{2k_1}_{ij},\,Q^{2k_2}_{lm}]=\delta_{lj}Q^{2(k_1+k_2)}_{im} 
-\delta_{im}Q^{2(k_1+k_2)}_{lj}
+\delta_{jm}Q^{2(k_1+k_2)}_{li}-\delta_{il}Q^{2(k_1+k_2)}_{jm}+\\
+r_i\delta_{im}Q^{2(k_1+k_2-1)}_{lj}
-r_j\delta_{lj}Q^{2(k_1+k_2-1)}_{im}+r_i\delta_{il}Q^{2(k_1+k_2-1)}_{jm} 
-r_j\delta_{jm}Q^{2(k_1+k_2-1)}_{li},  
\end{multline}
\begin{equation} 
\label{q2}
[Q^{2k_1}_{ij},\,Q^{2k_2-1}_{l}]=\delta_{lj}Q^{2k_1+2k_2-1}_{i}
-\delta_{il}Q^{2k_1+2k_2-1}_{j}-r_j\delta_{lj}Q^{2k_1+2k_2-3}_{i}+r_i\delta_{il}Q^{2k_1+2k_2-3}_{j},
\end{equation}
\begin{equation} 
\label{q3}
[Q^{2k_1-1}_{i},\,Q^{2k_2-1}_{j}]=Q^{2(k_1+k_2-1)}_{ij},\qquad\qquad
[Q^{2k_1-1}_{i},\,Q^{2k_2-1}_{i}]=0.   
\end{equation}

Since $Q^1_i=Q_i$ and $Q^{2k}_{ij}=-Q^{2k}_{ji}$, 
from~\eqref{q1},~\eqref{q2},~\eqref{q3} we obtain that the elements 
\begin{equation}
\label{qlnelem}
Q^{2k-1}_l,\qquad\quad Q^{2k}_{ij},\qquad\quad i,j,l\in\{1,\dots,n\},\qquad\quad 
i<j,\qquad\quad k\in\zsp,
\end{equation}
span the Lie algebra $L(n)$. 
From Lemma~\ref{lemq}
it follows that the elements~\eqref{qlnelem} are linearly independent 
over $\fik$ and, therefore, form a basis of $L(n)$. 

For $k\in\zsp$ set $L_{2k-1}=\lspan\big\{ Q^{2k-1}_l\,\big|\,l=1,\dots,n\big\}$ 
and $L_{2k}=\lspan\big\{ Q^{2k}_{ij}\,\big|\,i,j=1,\dots,n\big\}$. 
Here and below, for elements $v_1,\dots,v_s$ of a vector space, 
the expression $\lspan\{v_1,\dots,v_s\}$ denotes the linear span 
of $v_1,\dots,v_s$ over $\fik$. 

Then from~\eqref{q1}, \eqref{q2}, \eqref{q3} one gets 
$L(n)=\bigoplus_{i=1}^\infty L_i$ and $[L_i, L_j]\subset  L_{i+j}+ L_{i+j-2}$.
Thus the Lie algebra $L(n)$ is quasigraded (almost graded) 
in the sense of~\cite{quasigraded,skr-jmp}. Note that 
the algebra $L(n)$ is very similar to infinite-dimensional 
Lie algebras that were studied in~\cite{skr,skr-jmp}. 

It is easy to check that $Q_i$ satisfy 
relations~\eqref{rel1},~\eqref{rel2}, 
if we replace $p_i$ by $Q_i$ in these relations.  
Therefore, one has the homomorphism 
\begin{equation}
\label{is}
\vf\cl\mg(n)\to L(n),\quad\qquad \vf(p_i)=Q_i,\quad\qquad i=1,\dots,n.
\end{equation}

\begin{theorem} 
\label{gnL}
For all $n\ge 3$, the homomorphism~\eqref{is} is an isomorphism. 
Thus $\mg(n)$ is isomorphic to $L(n)$.
\end{theorem}
\begin{proof}
In the case $n=3$ this was proved in~\cite{ll}
for a different matrix representation of $L(3)$.

Define a filtration on $L(n)$ by vector subspaces $L^m\subset L(n)$ 
for $m\in\zp$ as follows 
\begin{equation*}
L^0=0,\qquad\quad
L^1=\lspan\{Q_1,\dots,Q_n\},\quad\qquad 
L^m=L^1+\sum_{i,j>0,\ i+j\le m}[L^i,L^j]\qquad\text{for}\,\ m>1. 
\end{equation*}
One has $L^m\subset L^{m+1}$ for all $m\in\zp$ and 
$L(n)=\bigcup_m L^m$. 

Since the elements~\eqref{qlnelem} are linearly independent, 
from~\eqref{q1},~\eqref{q2},~\eqref{q3} it follows that for all $q\in\zsp$ 
\begin{itemize}
\item the elements $Q^{2d-1}_l$, $Q^{2d}_{ij}$, 
$i,j,l\in\{1,\dots,n\}$, $i<j$, $1\le d\le q$,  
form a basis of $L^{2q}$,
\item $Q^{2d_1-1}_l$, $Q^{2d_2}_{ij}$, 
$i,j,l\in\{1,\dots,n\}$, $i<j$, $1\le d_1\le q$, 
$1\le d_2\le q-1$, form a basis of $L^{2q-1}$. 
\end{itemize}
This implies for all $m>0$ 
\begin{equation}
\label{L/L}
\dim\big(L^m/L^{m-1}\big)=
\left\{
\begin{array}{c}
n,\quad\text{if $m$ is odd},\\
{n(n-1)}/{2},\quad\text{if $m$ is even}.
\end{array}
\right.
\end{equation}

Consider a similar filtration on $\mg(n)$ by vector subspaces $\mg^m\subset\mg(n)$ 
\begin{equation*}
\mg^0=0,\qquad\quad
\mg^1=\lspan\{p_1,\dots,p_n\},\qquad\quad 
\mg^m=\mg^1+\sum_{i,j>0,\ i+j\le m}[\mg^i,\mg^j]\qquad\text{for}\,\ m>1.
\end{equation*}
Clearly, 
\begin{equation}
\label{gL}
\vf(\mg^m)=L^m\qquad\qquad\forall\,m\in\zp.
\end{equation} 
Combining~\eqref{gL} with~\eqref{L/L}, 
we see that it remains to prove for all $m>0$
\begin{equation}
\label{g/g}
\dim\big(\mg^m/\mg^{m-1}\big)\le
\left\{
\begin{array}{c}
n,\quad\text{if $m$ is odd},\\
{n(n-1)}/{2},\quad\text{if $m$ is even}.
\end{array}
\right.
\end{equation} 
Indeed, if~\eqref{g/g} holds then properties~\eqref{L/L},~\eqref{gL} 
imply that $\vf$ is an isomorphism. 

For $n=3$ the statement~\eqref{g/g} was proved in~\cite{ll}. 
Below we suppose $n\ge 4$. 
For $k\in\zsp$, set 
\begin{gather*}
P^{2k}_{ij}=(\ad p_i)^{2k-1}(p_j),\qquad\qquad 
i,j=1,\dots,n,\\
P^{2k-1}_{1}=(\ad p_2)^{2k-2}(p_1),\qquad\qquad
P^{2k-1}_{l}=(\ad p_1)^{2k-2}(p_l),\qquad l=2,3,\dots,n.
\end{gather*}

We will use the following notation for iterated Lie brackets 
of elements of $\mg(n)$  
\begin{equation}\label{iterated}
[e_1\,e_2\,\dots e_{s-1}\,e_s]=[e_1,[e_2,\dots,[e_{s-1},e_s]]\dots],
\qquad\quad e_1,\dots,e_s\in\mg(n).
\end{equation}
In such Lie brackets, 
for brevity we replace each $p_i$ by the corresponding index $i$. 
For example, 
\begin{gather}
\label{iijjklk}
[ii[jjk]lk]=[p_i,[p_i,[[p_j,[p_j,p_k]],[p_l,p_k]]]],\qquad\qquad 
P^{2k}_{ij}=[\ub{i\dots i}_{2k-1}j], \\
\label{PPPub} 
P^{2k-1}_{1}=[\ub{2\dots 2}_{2k-2}1],\qquad\qquad 
P^{2k-1}_{l}=[\ub{1\dots 1}_{2k-2}l],\qquad l=2,3,\dots,n. 
\end{gather}

For $V_1,V_2\in\mg(n)$ and $m\in\zp$, the notation 
\begin{equation}\label{modmgm}
V_1\equiv V_2\quad\mod\mg^m
\end{equation}
means that $V_1-V_2\in\mg^m$. 
The following lemma is proved in Section~\ref{sec_ap}.

\begin{lemma}[Section~\ref{sec_ap}] 
\label{lemma}
Let $n\ge 4$. 
Let $i,j,i',j'$ be distinct integers from $\{1,\dots,n\}$. 
Then for all $k_1,k_2\in\zp$ one has 
\begin{gather}
\label{PP0} 
[[\ub{i\dots i}_{2k_1}j][\ub{i\dots i}_{2k_2}j]]\equiv 0,\qquad  
\text{in particular},\quad 
[P^{2k_1+1}_j,P^{2k_2+1}_j]\equiv 0 
\mod\mg^{2k_1+2k_2+1},\\ 
\label{PP1}
P^{2(k_1+k_2+1)}_{ij}\equiv
 -P^{2(k_1+k_2+1)}_{ji}\quad\mod\mg^{2k_1+2k_2+1},\\ 
\label{PP2}
[P^{2k_1}_{ij},P^{2k_2+2}_{ij}]\equiv 0
\quad\mod\mg^{2k_1+2k_2+1}
\qquad\text{for}\quad k_1\ge 1,\\ 
\label{PP3}
[P^{2k_1+1}_i,P^{2k_2+1}_{j}]\equiv P^{2(k_1+k_2+1)}_{ij}
\quad\mod\mg^{2k_1+2k_2+1},\\
\label{PP4}
[P^{2k_1+1}_i,P^{2k_2+2}_{ij}]\equiv P^{2(k_1+k_2)+3}_{j}
\quad\mod\mg^{2k_1+2k_2+2},\\ 
\label{PP5}
[P^{2k_1+1}_{i},P^{2k_2+2}_{i'j'}]\equiv 0\quad\mod\mg^{2k_1+2k_2+2},\\ 
\label{PP6}
[P^{2k_1}_{ij},P^{2k_2+2}_{i'j'}]\equiv 0
\quad\mod\mg^{2k_1+2k_2+1}
\qquad\text{for}\quad k_1\ge 1,\\ 
\label{PP7}
[P^{2k_1}_{ij},P^{2k_2+2}_{ij'}]\equiv -P^{2(k_1+k_2+1)}_{jj'}
\mod\mg^{2k_1+2k_2+1}
\quad\text{for}\ 
k_1\ge 1.
\end{gather}
\end{lemma}

From Lemma~\ref{lemma}, by induction on $k\in\zsp$, we obtain that 
\begin{itemize}
\item the elements $P^{2d-1}_l$, $P^{2d}_{ij}$, 
$i,j,l\in\{1,\dots,n\}$, $i<j$, $1\le d\le k$, 
span the space $\mg^{2k}$,
\item 
$P^{2d_1-1}_l$, $P^{2d_2}_{ij}$, 
$i,j,l\in\{1,\dots,n\}$, $i<j$, $1\le d_1\le k$, 
$1\le d_2\le k-1$, span the space $\mg^{2k-1}$,
\end{itemize}
which implies~\eqref{g/g}.
\end{proof}

\begin{remark}
Clearly, formulas~\eqref{M},~\eqref{N} can be regarded as a ZCR 
with values in the Lie algebra~$L(n)$. 
Then formula~\eqref{M} becomes $M=\sum_{i=1}^n s^iQ_i$, 
where $Q_i\in L(n)$ is given by~\eqref{qie}. 
The homomorphism~\eqref{is} corresponds to this ZCR by Remark~\ref{remhomwea}.
\end{remark}

\section{Miura type transformations} 
\label{miura}

The definition of Miura type transformations (MTTs) was given 
in Section~\ref{detdesc}. 
In the present section we assume that all functions take values in~$\Com$.

Since the matrices~\eqref{M},~\eqref{N} form a ZCR for~\eqref{main}, 
the following system is compatible modulo~\eqref{main}
\begin{equation}\label{wxwtmn}
W_x=M^{\mathrm{T}}\cdot W,\qquad\qquad W_t=N^{\mathrm{T}}\cdot W,
\end{equation}
where $W=\big(w^1(x,t),\dots,w^{n+1}(x,t)\big)$ 
is a column-vector of dimension $n+1$ 
and $M^{\mathrm{T}}$, $N^{\mathrm{T}}$  
are the transposes of the matrices $M$, $N$ given by~\eqref{M},~\eqref{N}. 

Using~\eqref{M},~\eqref{N}, we see that equations~\eqref{wxwtmn} read  
\begin{gather}\label{vlas}
w^i_x=\la_is^iw^{n+1},\qquad\qquad i=1,\dots,n,\qquad\qquad
w^{n+1}_x=\sum_{j=1}^n\la_js^jw^i,\\
\label{vitlas}
\begin{split}
w^i_t=\la_iw^{n+1}\bigg(s^i_{xx}+s^i\Big(r_1+\la_1^2 
&+\frac12\langle S,RS\rangle+\frac32\langle S_x,S_x\rangle\Big)\bigg)+\\
&+\sum_{j=1}^n\la_i\la_jw^j\big(s^j_xs^i-s^i_xs^j\big),\qquad i=1,\dots,n,
\end{split}\\
\label{vn1tlas}
w^{n+1}_t=\sum_{j=1}^n\la_j w^j\bigg(s^j_{xx}+s^j\Big(r_1+\la_1^2 
+\frac12\langle S,RS\rangle +\frac32\langle S_x,S_x\rangle\Big)\bigg).
\end{gather}
Here $\la_1,\dots,\la_n\in\Com$ are parameters satisfying~\eqref{curve}.
In this section we assume $\la_i\neq 0$ for all~$i$.

To construct MTTs for~\eqref{main},
we are going to use some reduction of system~\eqref{vlas},  
\eqref{vitlas}, \eqref{vn1tlas}.
Equations~\eqref{vlas}, \eqref{vitlas}, \eqref{vn1tlas} imply  
\begin{equation}\label{pdvv0}
\frac{\pd}{\pd x}\bigg(\big(w^{n+1}\big)^2-\sum_{i=1}^n\big(w^i\big)^2\bigg)=0,
\qquad\qquad
\frac{\pd}{\pd t}
\bigg(\big(w^{n+1}\big)^2-\sum_{i=1}^n\big(w^i\big)^2\bigg)=0.
\end{equation}
Therefore, we can impose the constraint 
\begin{equation}\label{w2w2}
\big(w^{n+1}\big)^2=\sum_{i=1}^n\big(w^i\big)^2.
\end{equation}

Set 
\begin{equation}\label{wvi} 
\mtf^i={w^i}/{w^{n+1}},\qquad\qquad i=1,\dots,n. 
\end{equation}
From~\eqref{w2w2},~\eqref{wvi} one gets
\begin{equation}\label{vi21}
\sum_{i=1}^n\big(\mtf^i\big)^2=1.
\end{equation}

Since $\mtf^i=w^i/w^{n+1}$, 
one has $\mtf^i_x=w^i_x/w^{n+1}-\mtf^iw^{n+1}_x/w^{n+1}$.  
Combining this with~\eqref{vlas}, we get 
\begin{equation}\label{wixlas}
\mtf^i_x=\la_is^i-\mtf^i\sum_{j=1}^n\la_js^j\mtf^j,\qquad\qquad i=1,\dots,n.
\end{equation}
Similarly, using the formula $\mtf^i_t=w^i_t/w^{n+1}-\mtf^iw^{n+1}_t/w^{n+1}$ and equations~\eqref{vitlas},~\eqref{vn1tlas}, one obtains 
\begin{multline}\label{witlas}
\mtf^i_t=\la_is^i_{xx}+\sum_{j=1}^n\la_i\la_j\mtf^j\big(s^j_xs^i-s^i_xs^j\big)
-\mtf^i\sum_{j=1}^n\la_j\mtf^js^j_{xx}+\\
+\Big(r_1+\la_1^2 
+\frac12\langle S,RS\rangle +\frac32\langle S_x,S_x\rangle\Big)
\Big(\la_is^i-\mtf^i\sum_{j=1}^n\la_j\mtf^js^j\Big),\qquad
i=1,\dots,n.
\end{multline}

Using equations~\eqref{wixlas}, 
we want to express (at least locally) the functions $s^i$ 
in terms of $\la_j$, $\mtf^j$, $\mtf^j_x$.

Locally one can resolve the constraint $\sum_{j=1}^n (s^j)^2=1$ 
by taking $s^k=\sqrt{1-\sum_{i\neq k}(s^i)^2}$ 
for some $k\in\{1,\dots,n\}$. 
Here and below, we choose a suitable branch of the multivalued 
function $\sqrt{1-\sum_{i\neq k}(s^i)^2}$. 
For simplicity of notation, assume $k=n$. 
(The case $k\neq n$ can be studied analogously.) 
Then $s^n=\sqrt{1-\sum_{j=1}^{n-1}(s^j)^2}$.

Similarly, 
on a neighborhood of the point $\mtf^1=\mtf^2=\dots=\mtf^{n-1}=0$, $\mtf^n=1$, 
equation~\eqref{vi21} is equivalent to 
$\mtf^n=\sqrt{1-\sum_{j=1}^{n-1}(\mtf^j)^2}$, and system~\eqref{wixlas} becomes 
\begin{equation}\label{win1}
\mtf^i_x=\la_is^i-\mtf^i\sum_{j=1}^{n-1}\la_js^j\mtf^j-\mtf^i\la_n
\sqrt{1-\sum_{j=1}^{n-1}(s^j)^2}
\sqrt{1-\sum_{j=1}^{n-1}(\mtf^j)^2},\qquad\qquad i=1,\dots,n-1.
\end{equation}
Denote by $a^i=a^i(\la_1,\dots,\la_n,\mtf^1,\dots,\mtf^{n-1},s^1,\dots,s^{n-1})$ the right-hand side of~\eqref{win1}. 

For $\mtf^1=\dots=\mtf^{n-1}=0$ we have $\dfrac{\pd a^i}{\pd s^j}=\delta_{ij}\la_i$.
Recall that $\la_i\neq 0$. Therefore, 
by the implicit function theorem, 
on a neighborhood of the point $\mtf^1=\dots=\mtf^{n-1}=0$ from equations~\eqref{win1} we can express 
\begin{equation}
\label{uc-new}
s^i=R^i(\la_1,\dots,\la_n,\mtf^j,\mtf^j_x),\qquad\qquad i=1,\dots,n-1.
\end{equation} 
Combining~\eqref{uc-new} with the formula $s^n=\sqrt{1-\sum_{j=1}^{n-1}(s^j)^2}$, one gets 
\begin{equation}\label{snrj} 
s^n=
\sqrt{1-\sum_{j=1}^{n-1}\Big(R^j(\la_1,\dots,\la_n,\mtf^j,\mtf^j_x)\Big)^2}.
\end{equation}
Substituting~\eqref{uc-new},~\eqref{snrj} to~\eqref{witlas}, 
we obtain an evolution system of the form 
\begin{equation}
\label{wih}
\mtf^i_t=P^i(\la_1,\dots,\la_n,
\mtf^j,\mtf^j_x,\mtf^j_{xx},\mtf^j_{xxx}),\qquad\quad i=1,\dots,n,\qquad\qquad 
\sum_{i=1}^n(\mtf^i)^2=1.
\end{equation}
System~\eqref{wih} is connected with~\eqref{main} 
by the Miura type transformation~\eqref{uc-new},~\eqref{snrj}. 

Note that for system~\eqref{main} many solutions are known~\cite{ll-backl}. 
Therefore, it makes sense to describe 
how to construct solutions for~\eqref{wih} from solutions of~\eqref{main}.

Recall that~\eqref{wih} is obtained from~\eqref{vi21},~\eqref{wixlas},~\eqref{witlas} 
by eliminating $s^i$. 
Hence we need to describe solutions $\mtf^i$ 
of system~\eqref{vi21},~\eqref{wixlas},~\eqref{witlas} 
for a given solution $s^1,\dots,s^n$ of~\eqref{main}. 
We can use the fact that system~\eqref{vi21},~\eqref{wixlas},~\eqref{witlas} is obtained by 
the reduction~\eqref{w2w2},~\eqref{wvi} of~\eqref{wxwtmn}.

So let us fix 
a solution $S=\big(s^1(x,t),\dots,s^n(x,t)\big)$ of~\eqref{main}.
Then system~\eqref{wxwtmn} is compatible and is equivalent 
to a system of linear ordinary differential equations (ODEs). 
Indeed, one can first solve $W_x=M^{\mathrm{T}}\cdot W$ 
as an ODE with respect to~$x$, treating~$t$ as a parameter. 
Then one can substitute the obtained solution 
to the equation $W_t=M^{\mathrm{T}}\cdot W$, 
which is an ODE with respect to~$t$.

Suppose that the functions $s^i(x,t)$ are defined 
on a neighborhood of a point $(x_0,t_0)$. 
For any $z_1,\dots,z_{n+1}\in\Com$ satisfying 
\begin{equation}
\label{z2z2z}
(z_{n+1})^2=\sum_{i=1}^n(z_i)^2,\qquad\qquad z_{n+1}\neq 0,
\end{equation} 
consider the solution $w^1,\dots,w^{n+1}$ 
of the linear system~\eqref{wxwtmn} with the initial condition 
$w^j(x_0,t_0)=z_j$. 

From~\eqref{pdvv0},~\eqref{z2z2z} 
it follows that $w^j$ obey~\eqref{w2w2}. 
Since $w^{n+1}(x_0,t_0)=z_{n+1}\neq 0$, one has $w^{n+1}(x,t)\neq 0$ 
on some neighborhood of $(x_0,t_0)$. 
Then $\mtf^i(x,t)$ given by~\eqref{wvi} 
satisfy~\eqref{vi21},~\eqref{wixlas}, \eqref{witlas}.

For example, suppose that $s^i$ 
are constant, i.e., do not depend on $x$, $t$.  
Then $S=(s^1,\dots,s^n)$ is a constant solution of~\eqref{main}.
Since $s^i_x=0$, from~\eqref{M},~\eqref{N} 
we see that equations~\eqref{wxwtmn} read  
\begin{gather}
\notag
W_x=\tilde MW,\qquad W_t=\tilde NW,\qquad 
\tilde M=\sum_{i=1}^ns^i\la_i\big(E_{i,n+1}+E_{n+1,i}\big),\qquad
\tilde N=\Big(r_1+\la_1^2+\frac12\langle S,RS\rangle\Big)\tilde M.
\end{gather}
Since $[\tilde M,\tilde N]=0$ and 
the matrices $\tilde M$, $\tilde N$ do not depend on $x$, $t$, one has 
$$
W=\mathrm{e}^{(x-x_0)\tilde M+(t-t_0)\tilde N}\cdot Z,
$$ 
where $Z=(z_1,\dots,z_{n+1})$ is a column-vector satisfying~\eqref{z2z2z}. 
The corresponding functions $\mtf^i(x,t)$ are given by~\eqref{wvi}, 
where $w^i$ are the components of the vector $W$.

\begin{remark}
\label{mttvector}
It is well known that 
vector field representations of the WE algebra of an evolution PDE 
often lead to B\"acklund transformations. 
Let us show that 
the Miura type transformations constructed above 
correspond to some vector field representations of the WE 
algebra of~\eqref{main}.

The constructed MTTs are determined by system~\eqref{wixlas},~\eqref{witlas}, which is compatible modulo~\eqref{main}.
Let ${\tilde{a}}^i(\la_l,\mtf^l,s^l)$
be the right-hand side of~\eqref{wixlas} 
and ${\tilde{b}}^i(\la_l,\mtf^l,s^l,s^l_x,s^l_{xx})$
be the right-hand side of~\eqref{witlas}. Set 
\begin{equation}\label{abvf}
A=\sum_{i=1}^n
{\tilde{a}}^i(\la_l,\mtf^l,s^l)\frac{\pd}{\pd v^i},
\qquad\qquad
B=\sum_{i=1}^n{\tilde{b}}^i(\la_l,\mtf^l,s^l,s^l_x,s^l_{xx})
\frac{\pd}{\pd v^i}.
\end{equation}
Then compatibility of system~\eqref{wixlas},~\eqref{witlas} is equivalent to 
the equation
\begin{equation}\label{zcrvf}
D_x(B)-D_t(A)+[A,B]=0,
\end{equation}
where $D_x$, $D_t$ are the total derivative operators corresponding 
to system~\eqref{main}.

Let $\mathfrak{D}$ be the Lie algebra of vector fields 
on the space~$\Com^n$ with coordinates $\mtf^1,\dots,\mtf^n$. 
That is, $\mathfrak{D}$ consists of vector fields of the form 
$\sum_{i=1}^n h^i(\mtf^1,\dots,\mtf^n)\dfrac{\pd}{\pd \mtf^i}$. 

Equation~\eqref{zcrvf} says that formulas~\eqref{abvf} can be regarded as a ZCR 
with values in~$\mathfrak{D}$. By Remark~\ref{remhomwea}, 
this ZCR determines a homomorphism from the WE algebra of~\eqref{main} 
to~$\mathfrak{D}$. The homomorphism is given by 
\begin{equation}\label{hompv}
p_0\mapsto 0,\quad\qquad p_{n+1}\mapsto 0,\qquad\quad
p_j\,\mapsto\, 
\la_j\frac{\pd}{\pd \mtf^j}-\la_j
\mtf^j\sum_{i=1}^n\mtf^i\frac{\pd}{\pd \mtf^i},
\qquad\quad j=1,\dots,n,
\end{equation}
where $p_0,p_1,\dots,p_{n+1}$ are the generators of the WE algebra 
described in Theorem~\ref{theor-wea}.
Note that the vector fields~\eqref{hompv} are tangent 
to the submanifold given by equation~\eqref{vi21}.

\end{remark}

\section{Proof of Lemma~\ref{lemma}}
\label{sec_ap}

We prove Lemma~\ref{lemma} by induction on $k_1+k_2$. 
For $k_1+k_2=0$ (that is, $k_1=k_2=0$) 
the statements of Lemma~\ref{lemma} 
follow easily  from~\eqref{rel1},~\eqref{rel2}. 
Let $m\in\zp$ be such that the statements~\eqref{PP0}--\eqref{PP7} 
are valid for $k_1+k_2\le m$.  
We must prove \eqref{PP0}--\eqref{PP7} for $k_1+k_2=m+1$. 

Below $l\in\{1,\dots,n\}$ is such that $l\neq i$, $l\neq j$. 
In what follows, the symbol ``$=$'' denotes equality in the usual sense, 
and the symbol ``$\equiv$'' is used in the sense of~\eqref{modmgm}. 

Also, we often use the following property. 
If $V_1\equiv V_2\mod\mg^{r}$ for some $r\in\zp$ and $V_1,V_2\in\mg(n)$, 
then $[V_3,V_1]\equiv[V_3,V_2]\mod\mg^{r+r'}$ for any $r'\in\zp$ and 
$V_3\in\mg^{r'}$.

\textbf{Proof of~\eqref{PP0}.} 
We continue to use the notation~\eqref{iterated},~\eqref{iijjklk},~\eqref{PPPub}  
for Lie brackets of elements of~$\mg(n)$. 
For example, according to this notation, $[iP^{2q+2}_{ij}]=[p_i,P^{2q+2}_{ij}]$ 
and $[P^1_iP^{2q+2}_{ij}]=[P^1_i,P^{2q+2}_{ij}]$.

By the induction assumption, for all $q\le m$ one has 
\begin{gather}
\label{ii2q2j}
[\ub{i\dots i}_{2q+2}j]=[iP^{2q+2}_{ij}]=
[P^1_iP^{2q+2}_{ij}]\equiv P^{2q+3}_j\mod\mg^{2q+2},\\
\notag
[ll\ub{i\dots i}_{2q}j]=[ll[P^1_iP^{2q}_{ij}]]\equiv 
[llP^{2q+1}_j]=[l[P^1_lP^{2q+1}_{j}]]
\equiv [lP^{2q+2}_{lj}]=[P^1_lP^{2q+2}_{lj}]\equiv 
P^{2q+3}_j\mod\mg^{2q+2}.  
\end{gather}
Since~\eqref{ii2q2j} is valid for any $i\neq j$, 
we have also $[\ub{l\dots l}_{2q+2}j]\equiv P^{2q+3}_j\mod\mg^{2q+2}$. 
Therefore,  
\begin{equation} 
\label{lij}
[ll\ub{i\dots i}_{2q}j]\equiv [\ub{i\dots i}_{2q+2}j]\equiv   
[\ub{l\dots l}_{2q+2}j]\equiv P^{2q+3}_j\mod\mg^{2q+2}\quad\qquad
\forall\,i\neq j\neq l\neq i,\qquad\forall\,q\le m. 
\end{equation}

Without loss of generality, we can assume $k_2\ge 1$ in~\eqref{PP0}. 
By the induction assumption, we have   
$[[\ub{i\dots i}_{2k_1}j]\ub{i\dots i}_{2k_2-2}j]\equiv 0
\mod\mg^{2k_1+2k_2-1}$. 
Using this and the Jacobi identity, one gets
\begin{multline}\label{liijk2}
[l[\ub{i\dots i}_{2k_1}j][l\ub{i\dots i}_{2k_2-2}j]]=
[l[[\ub{i\dots i}_{2k_1}j]l]\ub{i\dots i}_{2k_2-2}j]
+[ll[\ub{i\dots i}_{2k_1}j]\ub{i\dots i}_{2k_2-2}j]
\equiv \\
\equiv
-[l[l\ub{i\dots i}_{2k_1}j]
[\ub{i\dots i}_{2k_2-2}j]]\mod\mg^{2k_1+2k_2+1},
\end{multline}

Using~\eqref{lij} and~\eqref{liijk2}, we obtain  
\begin{multline}
\label{l1}
[[\ub{i\dots i}_{2k_1}j][\ub{i\dots i}_{2k_2}j]]\equiv 
[[\ub{i\dots i}_{2k_1}j][ll\ub{i\dots i}_{2k_2-2}j]]
=-[[l\ub{i\dots i}_{2k_1}j][l\ub{i\dots i}_{2k_2-2}j]]+
[l[\ub{i\dots i}_{2k_1}j][l\ub{i\dots i}_{2k_2-2}j]]=\\
=[[ll\ub{i\dots i}_{2k_1}j][\ub{i\dots i}_{2k_2-2}j]]
-[l[l\ub{i\dots i}_{2k_1}j][\ub{i\dots i}_{2k_2-2}j]]
+[l[\ub{i\dots i}_{2k_1}j][l\ub{i\dots i}_{2k_2-2}j]]\equiv\\
\equiv
[[ll\ub{i\dots i}_{2k_1}j][\ub{i\dots i}_{2k_2-2}j]]
-2[l[l\ub{i\dots i}_{2k_1}j][\ub{i\dots i}_{2k_2-2}j]]
\mod\mg^{2k_1+2k_2+1}. 
\end{multline}

Since, by~\eqref{lij}, 
$[ll\ub{i\dots i}_{2k_1}j]\equiv [\ub{i\dots i}_{2k_1+2}j]
\mod\mg^{2k_1+2}$, 
from~\eqref{l1} it follows that  
\begin{equation}
\label{iikkjj}
[[\ub{i\dots i}_{2k_1}j][\ub{i\dots i}_{2k_2}j]]\equiv 
[[\ub{i\dots i}_{2k_1+2}j][\ub{i\dots i}_{2k_2-2}j]]-2[l[l\ub{i\dots i}_{2k_1}j][\ub{i\dots i}_{2k_2-2}j]]\mod\mg^{2k_1+2k_2+1}.
\end{equation} 
If $k_2\ge 2$, 
applying the same procedure to the term 
$[[\ub{i\dots i}_{2k_1+2}j][\ub{i\dots i}_{2k_2-2}j]]$ 
in equation~\eqref{iikkjj}, one gets 
\begin{equation*}
[[\ub{i\dots i}_{2k_1}j][\ub{i\dots i}_{2k_2}j]]
\equiv  
[[\ub{i\dots i}_{2k_1+4}j][\ub{i\dots i}_{2k_2-4}j]]
-2[l[l\ub{i\dots i}_{2k_1}j][\ub{i\dots i}_{2k_2-2}j]]
-2[l[l\ub{i\dots i}_{2k_1+2}j][\ub{i\dots i}_{2k_2-4}j]]
\mod\mg^{2k_1+2k_2+1}.
\end{equation*}
Thus, applying this procedure several times to the first summand of the right-hand side, we obtain 
\begin{equation}
\label{ijl}
[[\ub{i\dots i}_{2k_1}j][\ub{i\dots i}_{2k_2}j]]\equiv 
[[\ub{i\dots i}_{2(k_1+k_2)}j]j]
-2\sum_{s=1}^{k_2} [l[l\ub{i\dots i}_{2(k_1+s-1)}j][\ub{i\dots i}_{2(k_2-s)}j]]
\mod\mg^{2k_1+2k_2+1}.
\end{equation}
By the induction assumption and~\eqref{lij}, one has for all $s=1,\dots,k_2$
$$
[[l\ub{i\dots i}_{2(k_1+s-1)}j]i]\equiv 
[[P^1_lP^{2(k_1+s)-1}_j]P^1_i]\equiv [P^{2(k_1+s)}_{lj}P^1_i]\equiv 0
\mod\mg^{2(k_1+s)}.
$$
Therefore, 
\begin{equation}
\label{ijl1}
[l[l\ub{i\dots i}_{2(k_1+s-1)}j][\ub{i\dots i}_{2(k_2-s)}j]]\equiv 
[l\ub{i\dots i}_{2(k_2-s)}[l\ub{i\dots i}_{2(k_1+s-1)}j]j]
=-[l\ub{i\dots i}_{2(k_2-s)}jl\ub{i\dots i}_{2(k_1+s-1)}j]
\mod\mg^{2k_1+2k_2+1}.
\end{equation}
By the induction assumption and~\eqref{lij}, 
\begin{equation}
\label{verylong}
[l\ub{i\dots i}_{2(k_1+s-1)}j]\equiv [P^1_lP^{2(k_1+s)-1}_j]\equiv 
[P^{2(k_1+s)}_{lj}]
\equiv -[P^1_jP^{2(k_1+s)-1}_l]\equiv 
-[j\ub{i\dots i}_{2(k_1+s-1)}l]
\mod\mg^{2(k_1+s)-1}.
\end{equation}
Using~\eqref{ijl1},~\eqref{verylong}, and~\eqref{lij}, we obtain 
\begin{multline}
\label{longmult}
[l[l\ub{i\dots i}_{2(k_1+s-1)}j][\ub{i\dots i}_{2(k_2-s)}j]]\equiv
-[l\ub{i\dots i}_{2(k_2-s)}jl\ub{i\dots i}_{2(k_1+s-1)}j]
\equiv [l\ub{i\dots i}_{2(k_2-s)}jj\ub{i\dots i}_{2(k_1+s-1)}l]=\\
=[l\ub{i\dots i}_{2(k_2-s)}[jj\ub{i\dots i}_{2(k_1+s-1)}l]]
\equiv[l\ub{i\dots i}_{2(k_2-s)}[\ub{i\dots i}_{2(k_1+s)}l]]=
[l\ub{i\dots i}_{2(k_1+k_2)}l]
\mod\mg^{2k_1+2k_2+1}.
\end{multline}
Combining~\eqref{longmult} with~\eqref{ijl}, one gets 
\begin{equation}
\label{ijl2}
[[\ub{i\dots i}_{2k_1}j][\ub{i\dots i}_{2k_2}j]]\equiv 
-[j\ub{i\dots i}_{2(k_1+k_2)}j]
-2k_2[l\ub{i\dots i}_{2(k_1+k_2)}l]
\mod\mg^{2k_1+2k_2+1}\qquad
\forall\,k_1,k_2,\qquad k_1+k_2=m+1. 
\end{equation}
For $k_1=0$ and $k_2=m+1$, equation~\eqref{ijl2} implies 
\begin{equation}
\label{jii2m1}
[j\ub{i\dots i}_{2m+2}j]\equiv-(m+1)[l\ub{i\dots i}_{2m+2}l]\quad\mod\mg^{2k_1+2k_2+1}.
\end{equation}
Since we assume $n\ge 4$ in Lemma~\ref{lemma}, there is $b\in\{1,\dots,n\}$ such that 
$b\neq i$, $b\neq j$, $b\neq l$. Since~\eqref{jii2m1} is valid for any $i\neq j\neq l\neq i$, 
we get also 
\begin{equation}
\label{jii2m1b}
[j\ub{i\dots i}_{2m+2}j]\equiv-(m+1)[b\ub{i\dots i}_{2m+2}b],\qquad
[l\ub{i\dots i}_{2m+2}l]\equiv-(m+1)[b\ub{i\dots i}_{2m+2}b]\quad\mod\mg^{2k_1+2k_2+1}.
\end{equation}
Using~\eqref{jii2m1},~\eqref{jii2m1b}, one obtains
\begin{equation}
\label{m1bub}
(m+1)[b\ub{i\dots i}_{2m+2}b]\equiv-[j\ub{i\dots i}_{2m+2}j]
\equiv(m+1)[l\ub{i\dots i}_{2m+2}l]
\equiv -(m+1)^2[b\ub{i\dots i}_{2m+2}b]\mod\mg^{2k_1+2k_2+1}.
\end{equation} 
Equation~\eqref{m1bub} implies $[b\ub{i\dots i}_{2m+2}b]\equiv 0\mod\mg^{2k_1+2k_2+1}$. 
Combing this with~\eqref{ijl2},~\eqref{jii2m1b}, we obtain~\eqref{PP0}. 

\textbf{Proof of~\eqref{PP1}.} 
By the induction assumption and properties~\eqref{PP0},~\eqref{rel1},~\eqref{rel2}, 
\begin{gather*}
[\ub{i\dots i}_{2m+1}j]\equiv -[\ub{j\dots j}_{2m+1}i],\qquad 
[i\ub{j\dots j}_{2m}i]\equiv 0
\quad\mod\mg^{2m+1},\\ 
[[iij]j]\equiv 0\mod\mg^{3},
\quad\qquad [iij]\equiv [llj],\quad [jji]\equiv [lli] 
\mod\mg^{2},\quad\qquad [ilj]=0.
\end{gather*}
Using this, one gets 
\begin{multline*}
P^{2(k_1+k_2+1)}_{ij}=
[\ub{i\dots i}_{2m+3}j]=[ii\ub{i\dots i}_{2m+1}j]\equiv -[ii\ub{j\dots j}_{2m+1}i]
\equiv -[i[ij]\ub{j\dots j}_{2m}i]\equiv \\ 
\equiv -[[iij]\ub{j\dots j}_{2m}i]\equiv -[\ub{j\dots j}_{2m}[iij]i]
=[\ub{j\dots j}_{2m}iiij]
\equiv [\ub{j\dots j}_{2m}illj]=[\ub{j\dots j}_{2m}[il]lj]=\\
=[\ub{j\dots j}_{2m}[[il]l]j]=-[\ub{j\dots j}_{2m}jlli]
\equiv 
-[\ub{j\dots j}_{2m}jjji]=-P^{2(k_1+k_2+1)}_{ji}
\mod\mg^{2k_1+2k_2+1}.
\end{multline*}

\textbf{Proof of~\eqref{PP2}.} 
By the Jacobi identity and~\eqref{PP0}, 
\begin{multline}
\label{PP2p}
[P^{2k_1}_{ij},P^{2k_2+2}_{ij}]=
[[\ub{i\dots i}_{2k_1-1}j][\ub{i\dots i}_{2k_2+1}j]]
=[[[\ub{i\dots i}_{2k_1-1}j]i][\ub{i\dots i}_{2k_2}j]]+
[i[\ub{i\dots i}_{2k_1-1}j][\ub{i\dots i}_{2k_2}j]]=\\
=-[[\ub{i\dots i}_{2k_1}j][\ub{i\dots i}_{2k_2}j]]+
[i[\ub{i\dots i}_{2k_1-1}j][\ub{i\dots i}_{2k_2}j]]
\equiv[i[\ub{i\dots i}_{2k_1-1}j][\ub{i\dots i}_{2k_2}j]]
\mod\mg^{2k_1+2k_2+1}.  
\end{multline}
By the induction assumption and~\eqref{lij}, 
\begin{equation}
\label{ubidotsi2k}
[[\ub{i\dots i}_{2k_1-1}j][\ub{i\dots i}_{2k_2}j]]\equiv [P^{2k_1}_{ij}P^{2k_2+1}_{j}]\equiv 
P^{2(k_1+k_2)+1}_i
\mod\mg^{2k_1+2k_2}.
\end{equation}
Using~\eqref{PP2p},~\eqref{ubidotsi2k},~and~\eqref{PP0}, we obtain  
\begin{equation*}
[P^{2k_1}_{ij},P^{2k_2+2}_{ij}]\equiv[i[\ub{i\dots i}_{2k_1-1}j][\ub{i\dots i}_{2k_2}j]]
\equiv
[iP^{2(k_1+k_2)+1}_i]=[P^1_iP^{2(k_1+k_2)+1}_i]\equiv 0
\mod\mg^{2k_1+2k_2+1}. 
\end{equation*}

\textbf{Proof of~\eqref{PP3}.} 
By~\eqref{lij} and the induction assumption of~\eqref{PP5}, 
for any $q_1,\,q_2\in\zp$ such that $q_1+q_2\le m$, one has 
$[[\ub{l\dots l}_{2q_1}i][\ub{l\dots l}_{2q_2+1}j]]\equiv 
[P^{2q_1+1}_i,P^{2q_2+2}_{lj}]\equiv 0
\mod\mg^{2q_1+2q_2+2}$.
For $k_1+k_2=m+1$ this implies 
\begin{equation}
\label{llk1si}
[[\ub{l\dots l}_{2k_1+2s}i][\ub{l\dots l}_{2k_2-2s-1}j]]\equiv 0,\quad 
[[\ub{l\dots l}_{2k_1+2s+1}i][\ub{l\dots l}_{2k_2-2s-2}j]]\equiv 0
\mod\mg^{2k_1+2k_2}\qquad
\forall\,s,\quad 0\le s< k_2.
\end{equation}
Using~\eqref{llk1si} and the Jacobi identity, we get  
\begin{multline}\label{llk1ik2j}
[[\ub{l\dots l}_{2k_1}i][\ub{l\dots l}_{2k_2}j]]\equiv 
-[[\ub{l\dots l}_{2k_1+1}i][\ub{l\dots l}_{2k_2-1}j]]\equiv \\  
\equiv [[\ub{l\dots l}_{2k_1+2}i][\ub{l\dots l}_{2k_2-2}j]]\equiv\dots 
\equiv [[\ub{l\dots l}_{2(k_1+k_2)}i]j]=-[j\ub{l\dots l}_{2(k_1+k_2)}i] 
\mod\mg^{2k_1+2k_2+1}. 
\end{multline}
Combining~\eqref{llk1ik2j} with~\eqref{lij} and~\eqref{PP1}, one obtains 
\begin{multline*}
[P^{2k_1+1}_i,P^{2k_2+1}_{j}]\equiv [[\ub{l\dots l}_{2k_1}i][\ub{l\dots l}_{2k_2}j]]
\equiv -[j\ub{l\dots l}_{2(k_1+k_2)}i]\equiv \\ 
\equiv -[j\ub{j\dots j}_{2(k_1+k_2)}i]=
-P^{2(k_1+k_2+1)}_{ji}\equiv P^{2(k_1+k_2+1)}_{ij}
\mod\mg^{2k_1+2k_2+1}. 
\end{multline*}

\textbf{Proof of~\eqref{PP4}.} 
Consider first the case $k_1=0$. For $j\in\{1,\dots,n\}$, set 
$$
\tilde\jmath=
\left\{
\begin{array}{c}
1,\quad\text{if $j\neq 1$},\\
2,\quad\text{if $j=1$}.
\end{array}
\right.
$$
If $i=\tilde\jmath$ then 
$[P^{2k_1+1}_i,P^{2k_2+2}_{ij}]=
[i\ub{i\dots i}_{2k_2+1}j]=[\ub{\tilde\jmath\dots\tilde\jmath}_{2k_2+2}j]=
P_j^{2(k_1+k_2)+3}$ for $k_1=0$. 

Now suppose that $i\neq {\tilde\jmath}$ and $k_1=0$. By~\eqref{PP0}, 
\begin{equation}
\label{ijj2k2}
[i\ub{j\dots j}_{2k_2}i]\equiv 0\quad\mod\mg^{2k_2+1}. 
\end{equation}
From~\eqref{rel1} it follows that $[[ij]{\tilde\jmath}]=0$. 
Using~\eqref{PP1},~\eqref{lij},~\eqref{ijj2k2}, 
and $[[ij]{\tilde\jmath}]=0$, we obtain 
\begin{multline*}
[P^{2k_1+1}_i,P^{2k_2+2}_{ij}]\equiv -[P^{2k_1+1}_i,P^{2k_2+2}_{ji}]=
-[i\ub{j\dots j}_{2k_2+1}i]
\equiv -[[ij]\ub{j\dots j}_{2k_2}i]\equiv -[[ij]\ub{{\tilde\jmath}\dots {\tilde\jmath}}_{2k_2}i]=\\
=-[\ub{{\tilde\jmath}\dots {\tilde\jmath}}_{2k_2}[ij]i]=[\ub{{\tilde\jmath}\dots {\tilde\jmath}}_{2k_2}[iij]]
\equiv  
[\ub{{\tilde\jmath}\dots {\tilde\jmath}}_{2k_2}[{\tilde\jmath}{\tilde\jmath}j]]=P_j^{2(k_1+k_2)+3}
\mod\mg^{2k_1+2k_2+2}
\quad\text{for}\,\ k_1=0.  
\end{multline*}
If $k_1=0$ then 
$[P^{2k_1+1}_i,P^{2k_2+2}_{ij}]=[\ub{i\dots i}_{2m+4}j]$ for $m=k_1+k_2-1$.
We have proved~\eqref{PP4} for $k_1=0$, that is,  
\begin{equation}
\label{ubii2m2j}
[\ub{i\dots i}_{2m+4}j]\equiv P_j^{2m+5}\quad\mod\mg^{2m+4}
\qquad\qquad\forall\,i\neq j.
\end{equation}

Now consider the case $k_1\ge 1$. 
By~\eqref{lij}, for all $l\neq i$, $l\neq j$ one obtains
\begin{equation}
\label{p2k11ip2k2}
[P^{2k_1+1}_i,P^{2k_2+2}_{ij}]=-[P^{2k_2+2}_{ij},P^{2k_1+1}_i]\equiv
-[[\ub{i\dots i}_{2k_2+1}j][\ub{l\dots l}_{2k_1}i]]
\mod\mg^{2k_1+2k_2+2}.  
\end{equation}
By the induction assumption of~\eqref{PP5},
\begin{equation}
\label{ubii2k}
[[\ub{i\dots i}_{2k_2+1}j]l]=
[P^{2k_2+2}_{ij},P^{1}_l]\equiv 0\quad\mod\mg^{2k_2+2}.
\end{equation} 
Using~\eqref{lij},~\eqref{p2k11ip2k2},~\eqref{ubii2k}, we get  
\begin{equation*}
[P^{2k_1+1}_i,P^{2k_2+2}_{ij}]\equiv
-[[\ub{i\dots i}_{2k_2+1}j][\ub{l\dots l}_{2k_1}i]]\equiv 
-[\ub{l\dots l}_{2k_1}[\ub{i\dots i}_{2k_2+1}j]i]
=[\ub{l\dots l}_{2k_1}\ub{i\dots i}_{2k_2+2}j]\equiv 
[\ub{l\dots l}_{2m+4}j]\mod\mg^{2m+4}.  
\end{equation*}
Since~\eqref{ubii2m2j} is valid for all $i\neq j$, one has
$[\ub{l\dots l}_{2m+4}j]\equiv P_j^{2(k_1+k_2)+3}
\mod\mg^{2k_1+2k_2+2}$. 

\textbf{Proof of~\eqref{PP5}.} 
Since $n\ge 4$, there is $l\in\{1,\dots,n\}$ such that 
$l\neq i$, $l\neq i'$, $l\neq j'$. 

Consider first the case $k_1\ge 1$. 
Then $k_2\le m$. By the induction assumption of~\eqref{PP5}, 
$$
[[\ub{i'\dots i'}_{2k_2+1}j']l]=[P^{2k_2+2}_{i'j'},P^{1}_l]
\equiv 0,\qquad  
[[\ub{i'\dots i'}_{2k_2+1}j']i]=[P^{2k_2+2}_{i'j'},P^{1}_i]\equiv 0
\quad\mod\mg^{2k_2+2}. 
$$ 
Using this and~\eqref{lij}, one gets 
\begin{multline}
\label{PlP}
[P^{2k_1+1}_{i},P^{2k_2+2}_{i'j'}]=-[P^{2k_2+2}_{i'j'},P^{2k_1+1}_{i}]\equiv 
-[[\ub{i'\dots i'}_{2k_2+1}j']\ub{l\dots l}_{2k_1}i]
\equiv \\
\equiv
-[l\dots l[\ub{i'\dots i'}_{2k_2+1}j']i]\equiv 0
\mod\mg^{2k_1+2k_2+2}. 
\end{multline}
If we set $k_2=0$, $k_1=m+1$ 
then~\eqref{PlP} implies that for any distinct integers 
$c_1,c_2,c_3,c_4\in\{1,\dots,n\}$ 
\begin{equation}
\label{ccc'}
[[c_1c_2]\ub{c_4\dots c_4}_{2m+2}c_3]\equiv 0\quad\mod\mg^{2m+4}. 
\end{equation}
By~\eqref{lij}, one has $[\ub{c_4\dots c_4}_{2m+2}c_3]\equiv 
[\ub{c_2\dots c_2}_{2m+2}c_3]\mod\mg^{2m+2}$. 
Combining this with~\eqref{ccc'}, we obtain 
\begin{equation}
\label{ccc}
[[c_1c_2]\ub{c_2\dots c_2}_{2m+2}c_3]\equiv 0\quad\mod\mg^{2m+4}. 
\end{equation}
By the Jacobi identity, \eqref{ccc}, and~\eqref{lij},  
\begin{equation}
\label{ccc1}
[c_1\ub{c_2\dots c_2}_{2m+3}c_3]=
[[c_1c_2]\ub{c_2\dots c_2}_{2m+2}c_3]+[c_2c_1\ub{c_2\dots c_2}_{2m+2}c_3]
\equiv [c_2c_1\ub{c_2\dots c_2}_{2m+2}c_3]\equiv 
[c_2\ub{c_1\dots c_1}_{2m+3}c_3]\mod\mg^{2m+4}.
\end{equation}
Also, property~\eqref{PP1} implies 
\begin{equation}
\label{ccc2}
[c_1\ub{c_2\dots c_2}_{2m+3}c_3]\equiv -[c_1\ub{c_3\dots c_3}_{2m+3}c_2]
\quad\mod\mg^{2m+4}.
\end{equation} 

It remains to study the case $k_1=0$. 
Using~\eqref{ccc1} and~\eqref{ccc2}, for $k_1=0$ we get 
\begin{multline*}
[P^{2k_1+1}_{i},P^{2k_2+2}_{i'j'}]=
[i\ub{i'\dots i'}_{2k_2+1}j']\equiv 
[i'\ub{i\dots i}_{2k_2+1}j']\equiv
-[i'j'\dots j'i]\equiv -[j'i'\dots i'i]\equiv \\
\equiv[j'i\dots ii']\equiv
[ij'\dots j'i']\equiv -[ii'\dots i'j']=-[P^{2k_1+1}_{i},P^{2k_2+2}_{i'j'}]
\mod\mg^{2k_1+2k_2+2},\quad 
k_1=0. 
\end{multline*}
Therefore, 
$[P^{2k_1+1}_{i},P^{2k_2+2}_{i'j'}]\equiv 0\mod\mg^{2k_1+2k_2+2}$. 

\textbf{Proof of~\eqref{PP6}.} By~\eqref{PP5}, we have  
$[i,P^{2k_2+2}_{i'j'}]\equiv 0,\,\ [j,P^{2k_2+2}_{i'j'}]\equiv 0
\mod\mg^{2k_2+2}$. This implies~\eqref{PP6}. 

\textbf{Proof of~\eqref{PP7}.} 
By~\eqref{PP5}, 
$[[\ub{i\dots i}_{2k_1-1}j]j']\equiv 0\mod\mg^{2k_1}$. 
Using this and~\eqref{PP1},~\eqref{lij}, one obtains
\begin{multline*} 
[P^{2k_1}_{ij},P^{2k_2+2}_{ij'}]\equiv -[P^{2k_1}_{ij},P^{2k_2+2}_{j'i}]=
-[[\ub{i\dots i}_{2k_1-1}j]\ub{j'\dots j'}_{2k_2+1}i]
\equiv 
-[\ub{j'\dots j'}_{2k_2+1}[\ub{i\dots i}_{2k_1-1}j]i]=\\
=[\ub{j'\dots j'}_{2k_2+1}\ub{i\dots i}_{2k_1}j]\equiv
[\ub{j'\dots j'}_{2k_2+1}\ub{j'\dots j'}_{2k_1}j]
= P^{2(k_1+k_2+1)}_{j'j}\equiv 
-P^{2(k_1+k_2+1)}_{jj'}\mod\mg^{2k_1+2k_2+1}.
\end{multline*}

\section*{Acknowledgments}
The authors thank T.~Skrypnyk and V.~V.~Sokolov for useful discussions. 
Work of SI is supported by 
the Netherlands Organisation for Scientific Research (NWO) grants 639.031.515 and 613.000.906.    
JvdL is partially supported by the European Union through the FP6 Marie
Curie Grant (ENIGMA) and the European Science Foundation (MISGAM). 
SI is grateful to the Max Planck Institute for Mathematics (Bonn, Germany) 
for its hospitality and excellent working conditions during  02.2006--01.2007 and 06.2010--09.2010, 
when part of this research was done.

\end{document}